\documentclass[12pt, a4paper,reqno]{amsart}

\usepackage[T1]{fontenc}
\usepackage[utf8]{inputenc}

\usepackage{accents, amssymb}
\usepackage{xcolor}
\usepackage{calc}
\usepackage{graphicx}
\definecolor{bleu_sombre}{rgb}{0,0,0.6}  \definecolor{rouge_sombre}{rgb}{0.8,0,0}\definecolor{vert_sombre}{rgb}{0,0.6,0}
\usepackage[plainpages=false,colorl inks,linkcolor=bleu_sombre,
citecolor=rouge_sombre,urlcolor=vert_sombre,breaklinks]{hyperref}

\usepackage{subcaption}

\usepackage[english]{babel}
\usepackage{geometry}

\usepackage{amsmath,amssymb,amsthm,graphicx,amsfonts,enumerate,stmaryrd, mathrsfs}

\theoremstyle{plain}
\newtheorem{theorem}{{Theorem}}[section]
\newtheorem*{theorem*}{{Theorem}}
\newtheorem{proposition}[theorem]{Proposition}

\newtheorem*{proposition*}{Proposition}
\newtheorem{corollary}[theorem]{Corollary}
\newtheorem*{corollary*}{Corollary}
\newtheorem{lemma}[theorem]{Lemma}
\newtheorem{assumption}[theorem]{Assumption}
\newtheorem*{lemma*}{Lemma}

\theoremstyle{definition}

\newtheorem*{definition*}{Definition}

\theoremstyle{remark}
\newtheorem{remark}[theorem]{Remark}
\newtheorem{example}[theorem]{Example}

\makeatletter

\@addtoreset{equation}{section}
\makeatother

\renewcommand{\leq}{\leqslant}	\renewcommand{\geq}{\geqslant}

\newcommand{\R}{\mathbb{R}}	
\newcommand{\C}{\mathbb{C}}
\newcommand{\N}{\mathbb{N}}	

\newcommand{\dd}{\mathrm{d}}

\renewcommand{\Re}{\mathrm{Re}\,}
\renewcommand{\Im}{\mathrm{Im}\,}

\newcommand{\rond}[1]{\accentset{\circ}{#1}}
\newcommand{\Op}{\mathsf{Op}}

\title[Tunneling between magnetic wells]{Tunneling between magnetic wells\\ in two dimensions}

\author[S. Fournais]{S\o ren Fournais}
\address[S. Fournais]{Department of Mathematics, University of Copenhagen, Universitetsparken 5, DK-2100 Copenhagen \O, Denmark}
\email{fournais@math.ku.dk}

\author[Y. Guedes Bonthonneau]{Yannick Guedes Bonthonneau}
\address[Y. Guedes Bonthonneau]{Institut Galilée, Université Paris 13, Av Jean-Baptiste Clément 93430 - Villetaneuse}
\email{bonthonneau@math.univ-paris13.fr}

\author[L. Morin]{Léo Morin}
\address[L. Morin]{Department of Mathematics, University of Copenhagen, Universitets\-parken 5, DK-2100 Copenhagen \O, Denmark}
\email{lpdm@math.ku.dk }

\author[N. Raymond]{Nicolas Raymond}
\address[N. Raymond]{Univ Angers, CNRS, LAREMA, Institut Universitaire de France, SFR MATHSTIC, F-49000 Angers, France}
\email{nicolas.raymond@univ-angers.fr}

\begin{document}

	\begin{abstract}
The two-dimensional magnetic Laplacian is considered. 
We calculate the leading term of the splitting between the first two eigenvalues of the operator in the semiclassical limit under the assumption that the magnetic field does not vanish and has two symmetric magnetic wells with respect to the coordinate axes. This is the first result of quantum tunneling between purely magnetic wells under generic assumptions. The proof, which strongly relies on microlocal analysis, reveals a purely magnetic Agmon distance between the wells. Surprisingly, it is discovered that the exponential decay of the eigenfunctions away from the magnetic wells is not  crucial to derive the tunneling formula. The key is a microlocal exponential decay inside the characteristic manifold, with respect to the variable quantizing the classical center guide motion.
	\end{abstract}	
	
\maketitle

\tableofcontents

\section{Introduction and statement of the result}

\subsection{The tunneling effect}
This article considers the spectral theory of the magnetic Laplacian in two dimensions. It focuses on the description of a fundamental phenomenon arising in quantum physics, which has no classical counterpart, namely the tunneling effect. 
The standard quantity measuring this effect is the spectral gap between the first two eigenvalues in a double-well situation.

In the purely electric situation, i.e. when considering the Schrödinger operator $-h^2\Delta+V$ 
and assuming that the electric potential has two identical wells, mapped to each other by a symmetry,
the spectral gap between the first eigenvalues has been accurately calculated in the semiclassical limit $h\to 0$ in the seminal papers by Simon \cite{Simon84a, Simon84b} and by Helffer-Sjöstrand \cite{HS84, HS85b}. These authors have strongly contributed to our understanding of the eigenfunctions in the semiclassical limit. In particular, the localization properties near the minima of $V$ have been quantified thanks to the famous Agmon-Lithner estimates \cite{Agmon82} cast into the semiclassical regime, and very accurate WKB approximations have been established. These were key steps in the understanding of the exponential smallness of the eigenfunctions far from the minima -- the wells -- of $V$, that is where the interaction between the wells occurs. A paradigmatic example of a tunneling estimate in this context can be given in one dimension, say when $V$ is even and has non-degenerate minima at $\{\pm a\}$ with $a>0$. In this case, the lowest eigenvalues group in pairs that are exponentially close, each pair being separated from the others by gaps of order $h$. For instance, we have (see the original article \cite{Harrell} or the pedagogical articles \cite{Robert, BHR17}):
\[\lambda_2(h)-\lambda_1(h)=(1+o(1))ch^{\frac12}e^{-S/h}\,,\quad \mathrm{ with }\quad S=\int_{-a}^{a}\sqrt{V}\mathrm{d}x\,,\quad\mathrm{and }\quad c>0\,,\]
where the coefficient $c$ can be explicitly computed in terms of $V$. The key elements of the Helffer-Sjöstrand theory are described in the lecture notes \cite{Hel88} and in \cite{DS99}. 

\subsection{Tunneling in the presence of magnetic fields}
Shortly after their analysis of the purely electric tunneling effect, Helffer and Sjöstrand considered the case when a small magnetic field is added \cite{HS87b}. More precisely, they considered the operator $(-ih\nabla-b\mathbf{A})^2+V$ with $b$ small enough and with $V$ satisfying similar assumptions as in their previous works. This allowed them to establish similar tunneling formulas in weak magnetic fields but the general question of tunneling in magnetic fields remained open. Very recently, this situation has been reconsidered in \cite{FSW22} (in two dimensions) where the authors were able in a special setup to remove the smallness assumption on $b$ and to prove that the spectral gap does not close in the semiclassical limit. These authors consider the special case when $V$ has radial wells and the magnetic field is constant between them. Soon after followed \cite{HK22} where the exponential smallness of the gap was calculated (see also \cite{HKS24}). In \cite{Morin24}, the one-term asymptotics of the spectral gap was established by means of a concise and elementary presentation inspired by the original strategy of Helffer-Sjöstrand and by the key remark of \cite{FSW22} that in this example, the eigenfunctions coincide with Kummer special functions away from the wells.

However, the whole analysis in the latter works relied on the presence of $V$ and on a reduction of dimension due to radiality. This left open the purely magnetic case when $V=0$. 

Purely magnetic tunneling has been considered---still under an assumption of radial symmetry of the single-well problems---in the recent works \cite{FMR23,FM24,HK24}.

These works were preceded by a breakthrough understanding of magnetic tunneling in the case where the reduction of dimension was induced by a Neumann boundary condition (see \cite{BHR16,BHR22}). 
This reduction was based on the multiscale behavior of the eigenfunctions, with a stronger localization to the boundary than within the boundary itself. This gives rise to an effective Schrödinger-like operator along the boundary, which carries the tunneling effect.
This method was later succesfully applied to situations where the mechanism behind the dimensional reduction was slightly different (see \cite{FHK22,AA23}).
In the case we consider in this article, such a multiscale localization of the eigenfunctions is not true, and we cannot reduce to an effective one-dimensional Schrödinger operator.

\subsection{What do we know about magnetic wells?}\label{sub:history}

Surprisingly, the situation in two dimensions when $B$ is positive and has non-degenerate minima has been open up until rather recently. An important step in this direction was made in \cite{RVN15}, where the asymptotics of the lowest eigenvalues is related to the classical dynamics in a magnetic field, which is governed by the cyclotron and center guide motions. These Hamiltonian considerations, their quantization via Egorov theorems and semiclassical Birkhoff normal forms led to the following result.
\begin{theorem}[\cite{HK11} \& \cite{RVN15}]\label{thm.RVN}
	Assume that $B$ is smooth and positive and that it has a unique minimum $b_0$ at $0$, which is non-degenerate\footnote{In this case, we say that $B$ has a "well" at $0$.}. Consider a smooth magnetic potential $\mathbf{A}$, i.e. a function such that $\nabla\times\mathbf{A}=B$ and let $\mathscr{L}_h=(-ih\nabla-\mathbf{A})^2$.
	
Let $n\geq 1$. For $h$ small enough, the $n$-th eigenvalue $\lambda_n(h)$ exists and belongs to the discrete spectrum of $\mathscr{L}_{h}$ and
\[\lambda_n(h)=b_0h+((2nd_0+d_1)h^2+o(h^2)\,,\]
where
\[d_0=\frac{\sqrt{\det H}}{b_0}\,,\quad d_1=\frac{(\mathrm{Tr} H^{\frac12})^2}{2b_0}\,,\quad H=\frac12\mathrm{Hess} B(0)\,.\]
\end{theorem}
In \cite{RVN15}, it is established that the eigenfunctions are microlocalized in phase space near $(x,\xi)=(0,0)$ at the scale $h^{\frac12}$ (say if $\mathbf{A}(0)=0$). In particular, their space localization near the minimum of $B$ is isotropic contrary the multiscale cases evoked in the previous section. Note also that these localization estimates are not exponential estimates \emph{à la} Agmon: they only give that the eigenfunctions are $\mathcal{O}(h^\infty)$ away from $D(0,h^{\frac12-\eta})$, with $\eta>0$. By working a little more, we can prove a decay of order $\mathscr{O}(e^{-C/\sqrt{h}})$, see \cite[Section 1.4]{GBRVN21}. This decay can be used to get a rough estimate of tunneling when $B$ has a symmetric double well (see the discussion in \cite[Prop. 3.1 \& 3.2]{FMR23}):
\begin{equation}\label{eq.apriorigap}
\lambda_2(\mathscr{L})-\lambda_1(\mathscr{L})=\mathscr{O}(h^\infty)\,.
\end{equation}
Such rough estimates were already known in \cite[Remark 3.4]{HM96}, where tunneling between pure magnetic wells is mentioned as an open problem.

Under analyticity assumptions on $B$, it has been possible to prove a quasi-optimal exponential decay of the eigenfunctions away from the minimum of $B$ in \cite{GBRVN21} by means of the Fourier-Bros-Iagolnitzer transform, but without exhibiting an explicit and natural Agmon distance.  Let us describe this estimate. Let $M>0$. There exist $h_0,C>0$ such that, for all $h\in(0,h_0)$, any eigenfunction $\psi$ associated with $\lambda\leq b_0h+Mh^2$ satisfies
\[\int_{\R^2}e^{2d(x)/h}|\psi(x)|^2\mathrm{d}x\leq C\|\psi\|^2\,,\]
where $d$ is a smooth function quadratic near $0$. Obtaining such an estimate was motivated by the question of the optimality of WKB constructions done in \cite{BR20, GBNRVN21} in the case of a non-degenerate minimum of $B$. 

The case when $B$ is radial is considered in \cite{GBNRVN21}. In this case, we can show that the groundstate is radial and coincides with the groundstate of an electric Schrödinger operator. Thanks to the Helffer-Sjöstrand theory, this shows that the groundstate is approximated by the WKB Ansatz. This key property, reducing the analysis to the electric situation, and the observation in \cite{FSW22} about the Kummer functions motivated the article \cite{FMR23}, which establishes the first tunneling formula between pure radial magnetic wells. 

\subsection{Main result: a generic magnetic tunneling formula}
The present article aims at going beyond the radiality assumption. In this case we do not dispose of sufficiently good approximations of the eigenfunctions since we can neither use explicit special functions nor suitable WKB approximations inherited from the electric situation.

We consider the $2$-dimensional magnetic Laplacian
\[\mathscr{L}=(-ih\nabla-\mathbf{A})^2\,,\]
on $L^2({\mathbb R}^2)$ and denote the coordinates in ${\mathbb R}^2$ by $q=(q_1,q_2)$.
Due to the gauge invariance, we may assume that $\mathbf{A}$ has the form
\begin{equation}\label{the.gauge}
\mathbf{A}=(0,A_2)\,,\quad A_2(q)=\int_0^{q_1}B(s,q_2)\mathrm{d}s\,,
\end{equation}
where $B(q_1,q_2)$ denotes the magnetic field.
We work with a magnetic field that has precisely two non-degenerate minima. We assume symmetry of the magnetic field as well as partial analyticity in one of the coordinates (Assumptions~\ref{hyp.general} and~\ref{hyp.reg}). Furthermore, there is an important assumption on the complex level curve connecting the two minima (Assumption~\ref{hyp.complexconnection}). Finally, we also impose the technical Assumption~\ref{hyp.odd}.

More precisely the statements are as follows.

\medskip

\begin{minipage}{0.3\textwidth}
	\includegraphics[height=9cm]{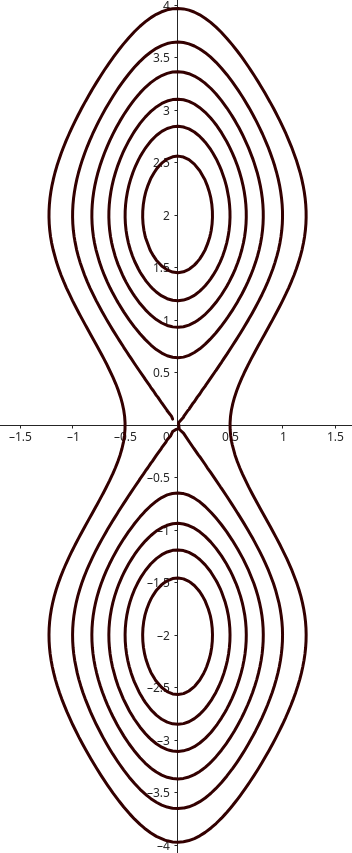}
\end{minipage}
\begin{minipage}{0.6\textwidth}
\begin{assumption}\label{hyp.general}
We consider a smooth magnetic field $B=B(q_1,q_2)$ with exactly two minima, which are positive and not attained at infinity.
We assume that $B$ is symmetric under reflection in the two coordinate axis,
\[B(-q_1,q_2)= B(q_1,-q_2)=B(q)\,, \quad \forall q \in \R^2.\]
We also assume that the two minima are $C_u=(0,c_u)$ and $C_d = (0,c_d)$ with $c_u>0$ and $c_d = -c_u$, that $B(C_u) := b_0 >0$, and that $C_u$ is a non-degenerate minimum, that is
\[\mathrm{Hess}_{C_u} B>0\,.\]
\end{assumption}
Note that the symmetry assumptions imply that $\mathrm{Hess}_{C_u} B$ is diagonal. Our choice of gauge is compatible with the symmetry of $B$ in the sense that
\[U\mathscr{L}=\mathscr{L} U\,,\]
where $U f(q)=f(-q)$.

\end{minipage}

\medskip

Let us now describe our other general assumptions.
\begin{assumption}[Partial analyticity and mild variations]\label{hyp.reg}
There exist $r,\eta_0>0$ such that the following holds. 
\begin{enumerate}[---]
\item The functions $q_1\mapsto B(q_1,q_2)$ and $q_1\mapsto \partial_2 A_2(q_1,q_2)$ are analytic in the strip $U_r:=\{q_1\in\C : |\Im q_1|< r\}$, uniformly in $q_2$, and their complex extensions belong to the symbol class \[S(U_r \times \R)=\{f\in\mathcal{C}^\infty(U_r\times\R) : \forall\alpha\in\N^2\,, \partial^\alpha f\in L^\infty(U_r\times\R)\}\,.\]
\item The complex extension of $B$ has mild variations in the sense that, for some $\varepsilon\in(0,1)$, we have, for all $q\in U_r\times\R$,
\begin{equation}\label{hyp.bound2}
	|B(q_1,q_2)-b_0|\leq \varepsilon b_0\,,\quad  | \partial_2 B(q_1,q_2) | \leq \frac{(1-\varepsilon)b_0}{2r}\,.
\end{equation}
\end{enumerate}
\end{assumption}

\begin{remark} Assumption \ref{hyp.reg} implies, for all $q\in U_r\times\R$,
\begin{equation}
\Re B(q_1,q_2)\geq (1-\varepsilon) b_0\,.
\end{equation}
\end{remark}

Then, we assume that there exists a unique complex path connecting the minima $C_u$ and $C_d$, which are disconnected in the real plane.

\begin{assumption}[Complex path connecting the minima]\label{hyp.complexconnection}
	There exists a smooth function $\gamma : \R\to\R$ such that $\|\gamma\|_\infty < r$ and
	\[B(i \gamma(q_2),q_2)=b_0\,,\quad \gamma(c_u)=\gamma(c_d)=0\,,\quad  \gamma'(c_d)>0\,,\quad \gamma'(c_u)<0\,.\]
	Moreover, $B(q_1,q_2)-b_0$ vanishes only on the curves $q_1 = \pm i \gamma(q_2)$, and in a non-degenerate way:
	\[ \partial_1 B(i\gamma(q_2),q_2) \neq 0, \quad \forall q_2 \in \R \setminus \lbrace c_d,c_u \rbrace. \]
\end{assumption}

This assumption turns out to be crucial since $\gamma$ appears in our tunneling formula. It is a natural assumption in the following sense. Since ${\rm Hess}_{C_u} B$ is diagonal and non-degenerate, there must always be exactly two complex curves along which $B-b_0$ vanishes on a small neighborhood of the wells. Due to the $q_1$-symmetry we also have $B(q_1,q_2) = b_0 \Leftrightarrow B(-\overline{q_1},q_2)=b_0$, and therefore such a curve must be purely imaginary. We simply assume that these curves defined on a neighborhood of $c_u$ and $c_d$ match each other following a global curve $q_1 = \pm i \gamma(q_2)$.

\begin{remark}
 Note that the function $\gamma$ is even due to the horizontal symmetry $B(q_1,q_2) = B(q_1,-q_2)$. Moreover, Assumption \ref{hyp.complexconnection} implies 
 \begin{equation}\label{eq.idB<0}
 i \partial_1 B (i \gamma(q_2),q_2) \gamma (q_2) < 0, \quad q_2 \neq c_u, c_d\,.
 \end{equation}
Indeed, since ${\rm Hess} _{C_u} B$ is diagonal and non-degenerate, one easily checks that \eqref{eq.idB<0} holds on a neighborhood of $c_u$. Since $i \partial_1 B$ does not vanish along the curve $\gamma$, it must stay negative for all $q_2$.
\end{remark}

Finally, we will also work under the following simplifying assumption.
\begin{assumption}\label{hyp.odd}
For all $(q_1,q_2)\in U_r\times\R$, we have
\[\Re q_1\geq 0\Longrightarrow\Re\frac{\partial_{1}B}{B}(q_1,q_2)\geq 0\,.\]	
\end{assumption}

For real $q_1$, Assumption \ref{hyp.odd} tells us that $B$ increases away from the vertical axis. In fact, this assumption could probably be removed using additional microlocal estimates in the proof of Lemma \ref{lem.phiutilde}. However, since it significantly simplifies the proof and does not hide any other interesting features on the tunneling effect, we decide to work under Assumption \ref{hyp.odd}.

\medskip
Under the above assumptions we are able to calculate the tunneling.

\begin{theorem}\label{thm.main}
Under Assumptions \ref{hyp.general}, \ref{hyp.reg}, \ref{hyp.complexconnection} and \ref{hyp.odd}, there exists $c_0>0$ such that, in the semiclassical limit $h \to 0$,
\begin{equation}\label{eq.tunneling}
\lambda_2(\mathscr{L})-\lambda_1(\mathscr{L})=c_0 h^{\frac32} e^{-S/h} (1+o(1))\,,
\end{equation}
where 
\begin{equation}\label{eq:S}
S=\int_{c_d}^{c_u}\int_{0}^{\gamma(q_2)} B(it,q_2)\,\mathrm{d}t\,\mathrm{d}q_2>0\,.
\end{equation}
\end{theorem}
\begin{remark} Let us make some comments on Theorem \ref{thm.main}.
\begin{enumerate}[\rm (i)]
\item Tunneling between purely magnetic wells remained an important open problem left open after the works of Helffer-Sjöstrand in the 1980's. It is for example mentioned as an open problem in \cite[Remark 3.4]{HM96}. Theorem~\ref{thm.main} actually disproves their conjecture that the order of magnetic tunneling should be $\mathscr{O}(e^{-C/\sqrt{h}})$.
\item We deduce that, when $h$ is small enough, the lowest eigenvalue is simple.
\item We have an explicit but complicated formula for $c_0$, see \eqref{eq.c0b}.
\item The action $S$ is a flux of $B$ through a surface in the complex plane.
In the formula \eqref{eq:S} one can also interpret $S$ as a purely magnetic Agmon distance. This quantity has never appeared in the literature before, to the best of our knowledge.
\item The tunneling formula \eqref{eq.tunneling} is the first general formula measuring optimally the exponentially small gap between the eigenvalues of a Laplacian with pure magnetic wells. Thereby, it is the only general magnetic analogue  of the electric tunneling results \emph{à la} Helffer-Sjöstrand. We explain in Section \ref{sec.orga} below how strongly the proof of Theorem \ref{thm.main} deviates from the original strategy in \cite{HS84, HS85b}. In particular, our strategy allows to establish optimal tunneling results for a class of pseudodifferential operators in two dimensions.
\end{enumerate}
\end{remark}

\begin{example}\label{ex.example}
	As an example, we consider the following familly of magnetic fields which satisfy our assumptions. Let $b_0>0$ and define, for $\varepsilon_1,\varepsilon_2\in(0,1]$,
	\[B(q)=b_0(1+\varepsilon_1g(q_1)+\varepsilon_2\langle q_1\rangle^{-2}W(q_2))\,,\] 
	with
	\[g(q_1):=1-\langle q_1\rangle^{-2}\,,\]
	where $0\leq W \leq \frac12$ is a smooth, even function with a double minimum at $c_u>0$ and $c_d =-c_u$, which is non-degenerate, and with minimal value $0$. In this case, $A_2$ and $B$ satisfy Assumptions~\ref{hyp.general} and \ref{hyp.reg}, if $\varepsilon_2$ is small enough.
%
	
	Assume that $\varepsilon_1^{-1}\varepsilon_2\leq 1$. 
	We want to solve $B(q)=b_0$, that is 
$q_1^2=-\varepsilon_3W(q_2)$ with $\varepsilon_3=\varepsilon_1^{-1}\varepsilon_2$.
	We find the curves
	\[q_1=\pm i\sqrt{\varepsilon_3W(q_2)}\,.\]
	In particular, we have
	\[|\Im q_1|=|q_1|\leq \sqrt{\varepsilon_3}\sqrt{W(q_2)}\,.\]
	Then, the smooth function
	\[\gamma(q_2)=\begin{cases}
		\sqrt{\varepsilon_3W(q_2)} &\text{ when } q_2\in(c_d,c_u)\,,\\
		-\sqrt{\varepsilon_3W(q_2)} & \text{ else },
	\end{cases}\]
solves the equation $B(i\gamma(q_2),q_2)=b_0\,.$ Moreover, for $q_2>c_u$,	$\gamma'(q_2)=-\frac{\varepsilon_3W'(q_2)}{2\sqrt{\varepsilon_3W(q_2)}}$ and thus, when $q_2\to c_u$ with $q_2>c_u$,
	\[\gamma'(c_u)= -\sqrt{\varepsilon_3}\sqrt{\frac{W''(c_u)}{2}}<0\,.\]

We finally check that
Assumption~\ref{hyp.odd} is satisfied if $\varepsilon_2\leq\varepsilon_1$. Indeed, we have
\[\frac{\partial_{1}B}{B}(q_1,q_2)=\frac{2 a(q_2) q_1}{(1+q_1^2)(1+q_1^2-a(q_2))}\,,\quad\mbox{ with }\quad a(q_2)=\frac{\varepsilon_1-\varepsilon_2W(q_2)}{1+\varepsilon_1}\,.\]
Note that $a(\R)\subset[0,\varepsilon_1]$. Then, take $(q_1,q_2)\in U_r\times\R$ such that $\Re q_1\geq 0$. The sign of $\Re\frac{\partial_{1}B}{B}(q_1,q_2)$ is the same as that of
\[\begin{split}&\Re\left(q_1(1+\overline{q}_1^2)(1+\overline{q}_1^2-a)\right)\\
	&=x\Re\left((1+\overline{q}_1^2)(1+\overline{q}_1^2-a)\right)-y\Im\left((1+\overline{q}_1^2)(1+\overline{q}_1^2-a)\right)\,,\end{split}\]
with $q_1=x+iy$. We have
\[\begin{split}
	\Re\left((1+\overline{q}_1^2)(1+\overline{q}_1^2-a)\right)&=(1+x^2-y^2)(1-a+x^2-y^2)-4x^2y^2\,,\\
		\Im\left((1+\overline{q}_1^2)(1+\overline{q}_1^2-a)\right)&=-2ixy(2-a+2x^2-2y^2)\,.
	\end{split}\]
	We deduce that, if $r>0$ is small enough, $\Re\left(q_1(1+\overline{q}_1^2)(1+\overline{q}_1^2-a)\right)\geq 0$ and thus $\Re\frac{\partial_{1}B}{B}(q_1,q_2)\geq 0$.
\end{example}

\subsection{Organization, heuristics and strategy}\label{sec.orga}
Let us now describe our strategy. As observed in \cite{RVN15, MRVN23}, a key to describing the spectrum of the magnetic Laplacian is to use suitable coordinates in which the magnetic field is constant. Indeed, the structure of the low-lying spectrum  is related to the characteristic manifold $M:=\{(q,p) : p-A(q)=0\}$, which also determines the Hamiltonian dynamics at low energy. The manifold $M$ is symplectic since $\mathrm{d}p\wedge\mathrm{d}q_{|M}=B(q)\mathrm{d}q_1\wedge \mathrm{d}q_2\neq 0$. This suggests to use the diffeomorphism given by $ \iota(q)= (A_2(q_1,q_2),\, q_2)=x$ so that $\iota^{-1}_*(\mathrm{d}p\wedge\mathrm{d}q_{|M})=\mathrm{d}x_1\wedge \mathrm{d}x_2$. The analyticity properties of $\iota$ are analyzed in Section \ref{sec.iota}. This change of variable conjugates $\mathscr{L}$ to a differential operator whose principal symbol is
\[\mathcal{B}(x_1, x_2)^2 \xi_1^2 + \left( \xi_2-x_1 - \alpha(x_1,x_2 ) \xi_1 \right)^2\,,\]
with $\mathcal B = B \circ \iota^{-1}$.
This suggests to use the symplectic change of variable $y_1=x_1-\xi_2$, $\eta_1=\xi_1$, $y_2=x_2-\xi_1$, $\eta_2=\xi_2$, which can be quantized by the metaplectic Egorov theorem (see \cite[Theorem 18.5.9]{Hormander85}) so that $\mathscr{L}$ is unitarily conjugated to a pseudodifferential operator whose principal symbol is
\[\mathcal{B}(\xi_2+x_1, x_2+\xi_1)^2 \xi_1^2 + \left( x_1 + \alpha(\xi_2+x_1,x_2+\xi_1 ) \xi_1 \right)^2\,.\]
Since the first eigenvalues of the magnetic Laplacian are of order $\mathcal{O}(h)$, the eigenfunctions are microlocalized in $(x_1,\xi_1)$ near $(0,0)$ at the scale $h^{\frac12}$. Thus, we make the rescaling $(x_1,\xi_1)\mapsto (h^{\frac12}x_1,h^{\frac12}\xi_1)$ and we get the operator $h\mathscr{M}$ in Proposition~\ref{prop.Mh} (which was already established in \cite{MRVN23}). After a change of quantization in $x_1$, the principal part of the operator is (with the standard notation for the Weyl quantization):
\[
h\mathscr{M} \approx
h\Op^{w,2}_{h}\Op^{w,1}_{1}\left(\mathcal{B}(\xi_2, x_2)^2 \xi_1^2 + \left( x_1 + \alpha(\xi_2,x_2 ) \xi_1 \right)^2\right)\,,\]
which is quadratic in $(x_1,\xi_1)$.  In the variables $(x_1,\xi_1)$, we have a harmonic oscillator corresponding to the well-known cyclotron motion of classical dynamics. The variables $(x_2,\xi_2)$ represent the center guide motion, which carries the tunneling phenomenon.

The rest of the article is devoted to the study of the pseudodifferential operator $\mathscr{M}$. The fact that, from the very beginning of the analysis, one casts the problem into a pseudodifferential framework is not common at all in the literature about the magnetic Laplacian. This idea of combining a change of coordinates and a metaplectic transform might a priori seem the source of technical troubles. These canonical coordinates, which are adapted to the magnetic field, will in fact reveal an \emph{optimal} exponential decay of the eigenfunctions that is responsible for the tunneling effect and that it would be hard to describe in the original space coordinates.

In order to understand the interaction between the magnetic wells, one needs to establish a very accurate description of the eigenfunctions when there is only one well. To achieve this, we consider the one well operator $\mathscr{M}_u=\mathscr{M}+\Sigma(x_2)$ by sealing the \textit{down} well by simply adding an electric potential, see \eqref{eq.Mu}. Since the tunneling effect is an exponentially small effect, we must describe as accurately as possible the decay of the eigenfunctions away from the well and prove exponentially good approximations. For that purpose, we consider the conjugated operator $\mathscr{M}_u^\Phi:=e^{\Phi(x_2)/h}\mathscr{M}_ue^{-\Phi(x_2)/h}$, which is still a pseudodifferential operator with operator-valued symbol, whose principal symbol is
\[	M_0(x_2,\xi_2)=\mathcal B(\xi_2+i\Phi'(x_2),x_2)^2D^2_{x_1}+(x_1+ \alpha(\xi_2+i\Phi'(x_2),x_2)D_{x_1})^2 + \Sigma(x_2)\,,\]
see Proposition \ref{prop.symbol} and Lemma \ref{prop.symbolh12}. Note that the eigenvalues of the non-self-adjoint operator $M_0(x_2,\xi_2)$ are the Landau levels (see Section \ref{sec.harmonic})
\[(2n-1)\mathcal{B}(\xi_2+i\Phi'(x_2),x_2)+\Sigma(x_2)\,,\quad n\geq 1\,.\]
The behavior of the resolvent near $b_0$ is studied in Lemma \ref{lem.Riesz}. We emphasize that the conjugation by $e^{-\Phi/h}$ amounts to considering a \emph{complex} magnetic field $\mathcal B(\xi_2+i\Phi'(x_2),x_2)$. Whereas complex electric potentials are rather common in the literature (as in \cite{MRVN23}), the case with complex magnetic fields seems to be widely open (see \cite{KNR24}). The present article gives a new motivation to study this situation.

The reader might be surprised that only weights in $x_2$ (and not also in $x_1$) are considered. In the electric situation, the usual Agmon estimates are established in all (space) variables. As we will see, we will not need a decay in all the variables, but only in $x_2$ which is a phase space variable for the original operator $\mathscr{L}$.

Section \ref{sec.4} is devoted to the study of the ellipticity of the non-self-adjoint operator $\mathscr{M}_u^\Phi-z$, with $z$ close to $b_0$. The most important result is Proposition \ref{prop.functionalin}, which is the corner stone of the proof of Theorem \ref{thm.main}. This elliptic estimate relies on Assumption \ref{hyp.subsolution} on the weight $\Phi$ and it tells that $\Re \mathcal B(\xi_2+i\Phi'(x_2),x_2)+\Sigma(x_2)-b_0$ behaves quadratically near $x_2=c_u$, uniformly in $\xi_2$. The proof of Proposition \ref{prop.functionalin} relies on the parametrix construction for $\mathscr{M}^\Phi_h-z$ given in Section \ref{sec.para}. More precisely, we construct left and right approximated inverses for the augmented matrix \[\mathscr{P}:=\begin{pmatrix}
	\mathscr{M}_u^\Phi-z&	\mathscr{F}\\
	\mathscr{G}	&0
\end{pmatrix},\]
where $\mathscr{F}$ and $\mathscr{G}$ are suitable pseudodifferential operators such that\footnote{This fact is overlooked in \cite{AA23} devoted to magnetic fields vanishing on curves, what leads to a mistake, which is probably inspired by \cite{BHR22} where the principal symbol of $\mathscr{P}$ is selfadjoint.}  $\mathscr{F}^*\neq\mathscr{G}$. The principal symbol of $\mathscr{P}$ is bijective and is analyzed in Proposition~\ref{prop.P0z}. The parametrix estimates are given in Proposition \ref{prop.parametrix}. The right inverse will be used in Section \ref{sec.6} when constructing WKB-like Ansätze. The left inverse is the key to the proof of Proposition~\ref{prop.functionalin}: it allows to relate the ellipticity property of $\mathscr{M}^\Phi_h$ to that of $\Re \mathcal B(\xi_2+i\Phi'(x_2),x_2)+\Sigma(x_2)-b_0$, see Corollary \ref{cor.coercivity}. Let us underline that we require that $\Phi'\in S(\R)+h^{\frac12}S_{\frac12}(\R)$. The (rather bad) class $h^{\frac12}S_{\frac12}(\R)$ is crucial to prove exponentially good approximations of the eigenfunctions of $\mathscr{M}_u$ away from the bottom of the well\footnote{This class is introduced in \cite{DR24} and fills a gap in \cite{BHR22, AA23} where non smooth weights are implicitly used.}. Fortunately, in this class, there is a Fefferman-Phong inequality that provides us with the desired coercivity, see Section \ref{sec.proof}.

In Section \ref{sec.5}, we deduce from Proposition \ref{prop.functionalin} (and of Corollary \ref{cor.functionalin}) an optimal exponential decay of the groundstate of $\mathscr{M}_u$. The best possible weight $\Phi$ that we would like to consider should saturate the inequality of Assumption \ref{hyp.subsolution}. That is why we consider $\varphi_u$ a solution to the eikonal equation $\mathcal B\left(i\varphi_u'(x_2),x_2\right) + \Sigma(x_2) =b_0$, see Lemma \ref{lem.Brond=b0}. In fact, we will rather work with a bounded version of $\varphi_u$, which we denote by $\tilde\varphi_u$, and we will show that $(1-\eta)\tilde\varphi_u$ satisfies Assumption \ref{hyp.subsolution} for any small $\eta >0$, see Section \ref{sec.weights}. This is where Assumption \ref{hyp.odd} comes into play and simplifies the analysis\footnote{Note that this coercivity analysis strongly relies on the variations of the magnetic field itself and not on the properties of some model operator -- the Montgomery or de Gennes operators -- as in \cite{BHR22, AA23}. This part of the analysis is unfortunately missing in \cite{AA23}, see Remark \ref{rem.AA}.}, see \eqref{eq.monotone}.

This weight will be enough to establish the exponential decay estimates on the eigenfunctions in Theorem~\ref{thm.agmon}. A corollary is that
\[\lambda_2(\mathscr{L})-\lambda_1(\mathscr{L})=\mathscr{O}(e^{-(1-\kappa)S/h})\,,\quad \lambda_3(\mathscr{L})-\lambda_2(\mathscr{L})\geq ch^2\,,\]
for some $c>0$ and all $\kappa >0$. The exponential behavior is much better than what the usual Agmon estimates give, see Proposition~\ref{prop.quasimodesexp}. In Lemma \ref{lem.phiuhat}, we also consider a refined weight $\hat{\varphi}_u$ which coincides with $\tilde{\varphi}_u$ away from a region of size $h^{\frac12}$ about $c_u$. This weight is the reason why one needs elliptic estimates with $\Phi'\in S(\R)+h^{\frac12}S_{\frac12}(\R)$.

Section \ref{sec.6} is devoted to establishing a very accurate approximation of the groundstate of $\mathscr{M}_u$. We start by a WKB-like construction of a quasimode based on the right parametrix given in Proposition \ref{prop.parametrix} \eqref{eq.i}. This short way of constructing quasimodes deviates from the method used in \cite{BHR16, BHR22, AA23}, which does not exploit the parametrix  construction. We construct two sequences $(a_j)_{j\geq 0}$ and $(z_j)_{j\geq 0}$ such that, by letting 	\[a^{[J]}=a_0+h^{\frac12}a_1+\ldots+h^{\frac{J}{2}}a_J\,,\quad z^{[J]}=z_0+h^{\frac12}z_1+\ldots+h^{\frac{J}{2}}z_J\,,\]
we have
\[(\mathscr{M}_u^{\varphi_u}-z^{[J]})\mathscr{Q}^{[J]}_+[\chi(x_2)a^{[J]}(x_2)]=\mathscr{O}(h^{\frac{J+1}{2}})\,,\]
where $\chi$ is a suitable cutoff function and $\mathscr{Q}^{[J]}_+$ is a perturbation of $\mathscr{F}$. Here, the main ingredient is the stationary phase method. The first term in the stationary phase gives the eikonal equation. This estimate tells us that the WKB-like function $e^{-\varphi_u(x_2)/h}\mathscr{Q}^{[J]}_+[\chi(x_2)a^{[J]}(x_2)]$ is a quasimode of $\mathscr{M}_u$. This quasimode, which gives a relevant solution to the approximation problem in \cite[Section 6.1]{Helffer09}, is an alternative to the WKB-quasimode given in \cite{BR20} (only in the analytic case) in the original coordinates and that turns out to be useless to study tunneling effect, contrary to our original belief dictated by the electric situation. In Proposition \ref{prop.approx}, we prove that this function is an exponentially good approximation of the groundstate, what we do not know and do not need for the quasimode in \cite{BR20}. This is where Lemma~\ref{lem.phiuhat} and the refined class of weights are used.
  
In Section \ref{sec.7}, Theorem \ref{thm.main} is proved. With the help of the exponentially weighted estimates in Theorem \ref{thm.agmon}, we prove that the space spanned by the one-well groundstates is a very good approximation of $\ker(\mathscr{M}-\lambda_1)\oplus\ker
(\mathscr{M}-\lambda_2)$, for which we construct an orthonormal basis. This is the aim of Proposition \ref{prop.interaction}, from which we deduce that
\[\lambda_2-\lambda_1=2|w|+\mathscr{O}(e^{-(1-\kappa)2S/h})\,,\]
where $w$ is the interaction term given by
\[w=\langle (\mathscr{M}-\mu_1)f_u,f_d\rangle\,.\]
The asymptotic analysis of $w$ is done in Section \ref{sec.analysisw}, where the one-well groundstates are replaced by their WKB-like approximations from Theorem \ref{thm.WKB}. Then, we use the stationary phase theorem to estimate the action of the involved pseudodifferential operators acting on the WKB amplitudes. This reveals the one-term expansion of $w$ and concludes the proof.

Finally, we gathered in Appendix \ref{sec.calculs} the calculations needed to compute the coefficient $c_0$ in the tunneling formula \eqref{eq.tunneling}.

\section{A unitarily equivalent operator and its conjugate}\label{sec.3}

\subsection{New coordinates}\label{sec.iota}
Inspired by \cite{MRVN23}, we prefer to work in the new coordinates $x=\iota(q)$ which are adapted to the geometry of the magnetic field. They are defined through the diffeomorphism
$\iota : \R^2_q\to\R^2_x$  which reads
\begin{equation}\label{def.iota}
\iota(q)= (A_2(q_1,q_2),\, q_2)=x\,.
\end{equation}
In fact, this defines Darboux coordinates for the magnetic $2$-form, meaning that $\dd x_1 \wedge \dd x_2 = B(q) \dd q_1 \wedge \dd q_2$. We recall that we chose a gauge where $A_1=0$, see \eqref{the.gauge}. Note that these coordinates preserve the symmetries of the problem, in the sense that $\iota(-q)=-\iota(q)$ and $\iota(q_1,-q_2)=(A_2(q_1,q_2),-q_2)$. The wells are also preserved since
\[\iota(C_u)=(A_2(C_u),c_u)=(0,c_u)=C_u\] and by symmetry $\iota(C_d)=C_d$.

In these coordinates, the magnetic field is $\mathcal B (x) = B \circ \iota^{-1}(x)$.
The function $\mathcal B$ still satisfies the symmetries $\mathcal B (-x_1,x_2) = \mathcal B(x_1,-x_2)=\mathcal B(x)$. In particular, the Hessian of $\mathcal B$ at $C_u$ is still diagonal. 

We also define
\begin{equation}\label{def.alpha}
\alpha(x) = \partial_2 A_2 \circ \iota^{-1}(x).
\end{equation}
We notice that $\mathcal B$ and $\alpha$ satisfy Assumption \ref{hyp.reg}.

\begin{lemma} \label{lem.analytic2}
If $B$ satisfies Assumption \ref{hyp.reg}, then there exists $\tilde r >0$ such that the functions $x_1 \mapsto \mathcal B(x_1,x_2)$  and $x_1 \mapsto \alpha(x_1,x_2)$ have holomorphic extensions to $U_{\tilde r}$ for all $x_2 \in \R$. Moreover, these extensions belong to $S(U_{\tilde r} \times \R)$.
\end{lemma}

\begin{proof}
Let us consider, for all $q\in \Omega_r:=U_r\times\C$, the natural extension of $A_2$,
\[\tilde{A}(q_1;q_2):=\int_{[0,q_1]}B(s,q_2)\mathrm{d}s\,.\] Then $\tilde A(\cdot,q_2)$ is holomorphic in $U_r$ and $\partial_{q_1} \tilde A \neq 0$ because of the small variations of $B$ (Assumption \ref{hyp.reg}). Thus, the range of $\tilde A( \cdot , q_2)$ is open. In fact, we can show that it is a biholomorphism.

Let us consider the injectivity. We have
\begin{equation}\label{eq.taylor}
\tilde A(\tilde q_1;q_2)-\tilde A(q_1;q_2)=(\tilde q_1-q_1)\int_{0}^1B(q_1+t(\tilde q_1-q_1),q_2)\mathrm{d}t\,.
\end{equation}
We notice that, for all $q_1, \tilde q_1\in U_{r}$ and all $q_2\in\R$,
\[\left|\int_{0}^1B(q_1+t(\tilde q_1-q_1),q_2)\mathrm{d}t\right|\geq (1-\varepsilon) b_0\,.\]
Thus,
\[|\tilde A(\tilde q_1;q_2)-\tilde A(q_1;q_2)|\geq (1-\varepsilon) b_0|\tilde q_1-q_1|\,,\]
which shows the injectivity. Let us now explain why the range contains a strip $U_{\tilde r}$ about $\R$.  Thanks to \eqref{eq.taylor} we get, for all $q\in U_r \times \R$,
\[\tilde A(q_1;q_2)-\tilde A(\Re q_1;q_2)=i\Im q_1\int_{0}^1B(\Re q_1+it\Im q_1,q_2)\mathrm{d}t\,,\]
so that
\[\Im\left(\tilde A(q_1;q_2)-\tilde A(\Re q_1;q_2)\right)=\Im q_1\int_{0}^1\Re B(\Re q_1+it\Im q_1,q_2)\mathrm{d}t\,,\]
and, for all $q\in U_r \times \R$,
\[|\Im \tilde A(q_1;q_2)|\geq (1-\varepsilon) b_0|\Im q_1|\,.\]
This shows the existence of the desired strip $U_{\tilde r} \times \R$.

Therefore, by the implicit function Theorem we can find a smooth function $v$ on $U_{\tilde r} \times \R$ such that
\[ \tilde A ( v(x_1,x_2) , x_2) = x_1. \]
Taking the $x_2$-derivative we also have $\partial_{x_2} v \in S(U_{\tilde r} \times \R)$.  The function $v$ is the inverse of $\iota$, and therefore the function
\[ \mathcal B (x) = B(v(x_1,x_2),x_2) \]
has a holomorphic extension to $U_{\tilde r} \times \R$, which belongs $S(U_{\tilde r} \times \R)$. In the same way, we obtain the extension of $\alpha$.
\end{proof}

\begin{remark}\label{remark.positiverealpart}
Note that, for $q \in U_r \times \R$,
\[ \partial_2 A_2 (q)  = \int_0^{{\rm Re} \, q_1} \partial_2 B(s,q_2) \dd s + \int_{{\rm Re} \, q_1}^{q_1} \partial_2 B(s,q_2) \dd s. \]
Therefore,
\[ |{\rm Im}\, \partial_2 A_2(q) | = \Big| {\rm Im} \int_{{\rm Re} \,q_1}^{q_1} \partial_2 B(s,q_2) \dd s \Big| \leq r \| \partial_2 B \|_{\infty} .\]
By \eqref{hyp.bound2} and \eqref{def.alpha} we deduce that
\[ |{\rm Im}\, \alpha | <  b_0 \frac{1-\varepsilon}{2}.\]
In particular, the real part ${\rm Re} \, (\mathcal B+i \alpha) \geq b_0 \frac{1-\varepsilon}{2}$ is always positive.
\end{remark}

\subsection{Normal form}\label{sec.normal}

As underlined in \cite{RVN15} and \cite{MRVN23}, the magnetic Laplacian $\mathscr L$ turns out to be simpler understood as a pseudodifferential operator. In \cite{MRVN23} for instance, a metaplectic transformation which mixes position and momentum variables was used to find a normal form for the operator. Let us recall this result, which is the first key element of our analysis. 

We will denote coordinates in the phase space ${\mathbb R}^4$ by $X = (X_1,X_2)$, where $X_1=(x_1,\xi_1)$ and $X_2=(x_2,\xi_2)$. All through the article, we will use the Weyl quantization in the spirit of \cite{Zworski}. Given a symbol $n(x_1,\xi_1)$, its quantization is formally defined by
\[ \Op_h^{w,1}(n) u(x_1) = \frac{1}{2\pi h} \int_{\R^2} e^{\frac{i}{h}\xi_1 (x_1-y_1)} n \Big( \frac{x_1+y_1}{2} ,\xi_1 \Big) u(y_1) \dd y_1 \dd \xi_1.\]
The upper index $1$ refers to the variable $x_1$. When $h=1$, $\Op_1^{w,1}$ is the non-semiclassical quantization. We can also quantize operator-valued symbols. Assume that for all $X_2 =(x_2,\xi_2) \in \R^2$ we are given a bounded operator $N(X_2) \in \mathcal L( \mathscr B, \mathscr B')$ between Banach spaces $\mathscr B$, $\mathscr B'$. Also assume that $N(X_2)$ belongs to the class of bounded symbols
\begin{multline*}
 S(\R^2, \mathcal L( \mathscr{B}, \mathscr{B}')) \\= \Big\lbrace N \in \mathscr{C}^\infty(\R^2, \mathcal L(\mathscr B, \mathscr B')) ,\quad \sup_{X_2 \in \R^2} \| \partial^\gamma N (X_2) \| \leq C_\gamma, \quad \forall \gamma \in \N_0^2 \Big\rbrace. 
 \end{multline*}
Then the quantization of $N$ is formally defined by the same formula
\[ \Op^{w,2}_h (N) u(x_2) = \frac{1}{2\pi h} \int_{\R^2} e^{\frac ih \xi_2 (x_2 - y_2)}  N \Big( \frac{x_2+y_2}{2}, \xi_2 \Big) u(y_2) \dd y_2 \dd \xi_2.\]
This formula defines a bounded operator $L^2(\R, \mathscr{B}) \to L^2(\R, \mathscr B')$. In this article, $\mathscr{B}$ will be either equal to $L^2(\R_{x_1})$ or, for some $k \in \mathbb N$,
\begin{equation}
B^k(\R) = \lbrace f \in L^2(\R), \quad x^j \partial_x^{k-j} f \in L^2(\R), \quad 0 \leq j \leq k \rbrace.
\end{equation}
The space $B^k(\R)$ is given the natural norm $\| f \|_{B^k}^2 = \sum_{j=0}^k \| x^j \partial_x^{k-j} f \|^2_{L^2}$.
For functions of two variables $(x_1,x_2)$ we will often use the mixed norms
\begin{align}
\| f \|_{B^k \otimes L^2}^2 = \int \|f(\cdot,x_2)\|_{B^k}^2 \,\dd x_2.
\end{align}
For the properties of the Weyl quantization we refer to \cite{Zworski,Keraval}. In particular, when $n(x_1,x_2,\xi_1,\xi_2)$ is a symbol in two variables, we can consider the partial quantizations
\[ \Op_h^{w,1} (n), \quad \Op_h^{w,2}(n), \]
and also
\[\Op_h^{w,2} \Big( \Op_h^{w,1}(n) \Big), \]
which coincides with the standard Weyl quantization on $\R^2$. With these notations, we have the following proposition, taken from \cite[Proposition 2.2]{MRVN23}. It is based on the metaplectic Egorov theorem.

\begin{proposition}\label{prop.Mh}
	The operator $\mathscr{L}$ is unitarily equivalent to
	\[
	h \mathscr{M} := h \Op^{w,2}_{h} \Op_1^{w,1}
	m\,,
	\]
	where
\begin{align}\label{eq.m}
m(x,\xi) &=\mathcal{B}(\xi_2 + h^{1/2} x_1, x_2 + h^{1/2} \xi_1)^2 \xi_1^2 + \left( x_1 + \alpha(\xi_2 + h^{1/2}x_1,x_2 + h^{1/2}\xi_1) \xi_1 \right)^2\nonumber \\
			& \quad + h T(\xi_2 +
			h^{1/2}x_1, x_2 + h^{1/2}\xi_1)\,,
\end{align}
\end{proposition}

\begin{proof}[Idea of the proof]
Let us give some more details about the unitary equivalence in Proposition~\ref{prop.Mh}.
The first step is to implement the change of variables $x= \iota(q)$. This leads to the differential operator
\begin{align}
\widetilde{{\mathscr L}} = \Op_h^{w}\Big( {\mathcal B}(x_1,x_2)^2 \xi_1^2 + (\xi_2 - x_1 + \alpha(x_1,x_2) \xi_1)^2 + h^2 T \Big).
\end{align}
We then implement the symplectic change of variable $y_1=x_1-\xi_2$, $\eta_1=\xi_1$, $y_2=x_2-\xi_1$, $\eta_2=\xi_2$, replacing $\widetilde{\mathscr{L}}$ by the conjugated operator $\mathscr{U} \widetilde{\mathscr{L}} \mathscr{U}^\ast$, where $\mathscr{U}$ is the unitary operator
\begin{align*}
\mathscr{U}f(y) &= (2\pi h)^{-1} \iint e^{ih^{-1} \xi_2(y_2-x_2)} f(y_1+\xi_2,x_2)\, \dd x_2 \dd \xi_2  \\
&= 
(2\pi h)^{-2} \iint e^{ih^{-1} \xi(y-x) + i h^{-1}\xi_1 \xi_2 } f(x) \,\dd x \dd \xi.
\end{align*}
The final step is to implement the scaling $y_1 \mapsto \sqrt{h} y_1$ as a unitary map on $L^2$. We see that  $X_2 \mapsto \Op_1^{w,1} \big(m(\cdot, X_2)\big) $ belongs to the claimed class of symbols using the Calderón-Vaillancourt Theorem, and noticing that $m$ is quadratic with respect to $(x_1,\xi_1)$.
\end{proof}

\begin{remark}
Note that  $\mathscr{M}$ still commutes with the symmetry $U$,
\begin{equation}\label{eq.MU=UM}
\mathscr{M}U=U\mathscr{M}\,.
\end{equation}
\end{remark}

To main order, the operator $\mathscr{M}$ is a quadratic form in $X_1=(x_1,\xi_1)$ as one can guess by formal expansion in powers of $h^{\frac 12}$. This quadratic form has $X_2$-dependent coefficients, and it is equivalent to a harmonic oscillator with frequency $\mathcal B(\xi_2,x_2)$. This oscillator in $X_1$ can be interpreted as a quantization of the cyclotron motion, see \cite{RVN15}. The value of Proposition \ref{prop.Mh} is to decouple the variables $X_1$ and $X_2$, at least to main order. For this reason the operator $\mathscr M$ is simpler than $\mathscr L$. In the remainder of the article, we only work with $\mathscr M$, and prove a tunneling formula directly for this operator. As we explain below, tunneling holds in the variables $X_2$. 

\subsection{Exponential conjugation of the one-well operator}

We define the one-well operator by sealing the bottom well. Let $\Sigma\in\mathscr{C}^\infty_0(\R)$ be a non-negative bump function with arbitrarily small support, small $\| \Sigma \|_{\infty}$, and such that $\Sigma(c_d) >0$. We define the upper one-well operator by
\begin{equation}\label{eq.Mu}
\mathscr M_u = \mathscr M + \Sigma(x_2).
\end{equation}
The purpose of the next sections is to obtain optimal decay estimates on the eigenfunctions of $\mathscr{M}_u$, and WKB-like constructions. For this reason we consider the operator conjugated by an exponential weight. The proposition below shows that the conjugated operator is still a pseudodifferential operator.

\begin{proposition}\label{prop.symbol}
Consider a smooth real-valued function $\Phi\in S(\R)$ such that $|\Phi'|\leq \tilde r$, where $\tilde r$ is given by Lemma \ref{lem.analytic2}. Then the operator
\[\mathscr{M}_u^\Phi:=e^{\Phi(x_2)/h}\mathscr{M}_ue^{-\Phi(x_2)/h}\] is a pseudodifferential operator that can be written in the form $\Op^{w,2}_{h} \Op_1^{w,1} m_u^\Phi$ where 
\begin{equation}\label{eq.nPhi}
m_u^\Phi(x,\xi)=m(x_1,x_2,\xi_1,\xi_2+i\Phi'(x_2)) +\Sigma(x_2)+ \mathscr{O}(h^2)\,,
\end{equation}
in the topology of $S(\R^2,\mathcal{L}(B^2(\R), L^2(\R)))$.
Moreover, we have the following expansion for all $J \in \N$:
\begin{equation}\label{eq.Mhphi}
	\mathscr{M}_u^\Phi=\Op^{w,2}_h \big(M_0+h^{\frac12}M_1+hM_2+\ldots+h^{\frac{J}{2}}M_J \big)+ h^{ \frac{J+1}{2}}\mathscr{R}_{J,h}\,,
\end{equation}
where the symbol of $\mathscr{R}_{J,h}$ belongs to $S(\R^2,\mathcal{L}(B^{J+3}(\R), L^2(\R)))$, and where the first expressions of the $M_j$'s are given explicitly by 
\[\begin{split}
	M_0(x_2,\xi_2)&=\mathcal B(\xi_2+i\Phi'(x_2),x_2)^2D^2_{x_1}+(x_1+ \alpha(\xi_2+i\Phi'(x_2),x_2)D_{x_1})^2 + \Sigma(x_2)\,,\\
	M_1(x_2,\xi_2)&= \Op^{w,2}_h (m_1)\,,\\
	M_{2}(x_2,\xi_2)&=\Op^{w,2}_h (m_2)\,,
\end{split}\]
with
\[\begin{split}
	m_1&= 2\xi^2_1\mathcal B\nabla\mathcal B\cdot X_1+2\xi_1(x_1+\alpha\xi_1)\nabla \alpha\cdot X_1  \,,\\
	m_2&=\xi_1^2(\nabla\mathcal B\cdot X_1)^2+\xi_1^2 \mathcal B \nabla^2\mathcal B(X_1,X_1)+\xi^2_1(\nabla\alpha\cdot X_1)^2 \\ & \quad +\xi_1(x_1+\alpha\xi_1)\nabla^2\alpha(X_1,X_1)+ T(\xi_2+i\Phi'(x_2),x_2)\,,
\end{split}\]
and where all the coefficient-functions are evaluated at $(\xi_2+i\Phi'(x_2),x_2)$.
\end{proposition}

\begin{remark}
Notice that in order not to overload the symbols, we have omitted the dependence on $\Phi$ in the notation for the operators $M_j$.
\end{remark}

\begin{proof}
The fact that $\mathscr{M}^\Phi_u$ is pseudodifferential operator with the expansion \eqref{eq.nPhi} follows from the same considerations as in \cite[Section 3]{Nakamura} (see also \cite[Appendix~A]{DR24} where a simple proof is given). The expansion \eqref{eq.Mhphi} is then inspired by \cite[Section 3.2]{MRVN23}, and is proven as follows. The symbol $m_u^\Phi$ given by \eqref{eq.nPhi} and \eqref{eq.m} can be expanded in powers of $h^{\frac 12} (x_1,\xi_1)$ using the Taylor formula. At order $J >0$, the remainder $r_{J,h}$ is bounded by
\[ |r_{J,h}(X_1,X_2)| \leq C h^{\frac{J+1}{2}}  \langle |x_1| + |\xi_1| \rangle^{J+3}. \]
Similarly bounding the $x_1$-derivatives of $r_{J,h}$, we find that $\Op_1^{w,1} r_{J,h}$ defines a bounded operator in $\mathcal L( B^{J+3}(\R),L^2(\R))$, since the coefficients $\mathcal B$ and $\alpha$ are in $S(\R)$. Bounding $x_2$-derivatives we finally find that $\Op_1^{w,1} r_{h,J}$ is a uniformly bounded symbol in the class $h^{\frac{J+1}{2}} S(\R^2, \mathcal L( B^{J+3}(\R),L^2(\R)))$.
\end{proof}

There is an extension of Proposition \ref{prop.symbol} to slightly more general weights, involving the class of $h$-dependent functions
\begin{equation}
S_{\frac 12}(\R) = \big\lbrace \Phi ( \cdot ; h) \in \mathscr{C}^{\infty}(\R) \, ; \forall n, \quad |\partial_x^n \Phi(x;h)| \leq C_\alpha h^{- \frac n2} \big\rbrace.
\end{equation}
Such weights will be useful when proving accurate approximation of the one-well eigenfunctions.

\begin{lemma}\label{prop.symbolh12}
Consider a smooth bounded real-valued function $x_2\mapsto\Phi(x_2)$ with $\Phi'\in S(\R)+h^{\frac12}S_{\frac12}(\R)$ such that $|\Phi'|\leq \tilde r$. Then, 
$\mathscr{M}_u^\Phi:=e^{\Phi(x_2)/h}\mathscr{M}_ue^{-\Phi(x_2)/h}$ is a pseudodifferential operator that can be written in the form 
\[	\mathscr{M}_u^\Phi=\Op^{w,2}_h(M_0+h^{\frac12}M_1)+ h\mathscr{R}_h\,,\]
where the symbol of $\mathscr{R}_h$ belongs to
\begin{multline*}
	S_{\frac12,0}\big(\R^2,\mathcal{L}(B^4(\R),L^2(\R))\big)\\
=\big\{a\in\mathscr{C}^\infty : \forall\alpha\in\N^2\,,\exists C_\alpha>0 : \|\partial^\alpha_{X_2} a\|_{\mathcal{L}(B^4(\R),L^2(\R))}\leq C_\alpha h^{-\frac{\alpha_1}{2}}\big\}\,.
\end{multline*}
\end{lemma}

\begin{proof}
 Under the assumptions on $\Phi$, each time we differentiate $\mathcal B(\xi_2 + i \Phi'(x_2),x_2)$ with respect to $x_2$ we loose a power of $h^{\frac 12}$. Hence the same  holds for the symbol $m_u^\Phi$. Proceeding as in Proposition \ref{prop.symbol} with $J=1$, the result follows.
\end{proof}

\section{Functional inequalities} \label{sec.functionaline}\label{sec.4}
\subsection{Statements}
The inequality in the following Proposition~\ref{prop.functionalin} is the core of the Agmon-type decay estimates and of a WKB-like approximation of the groundstate of $\mathscr{M}_u$. It can be established under the following assumption on the weight $\Phi$. Let $\chi_0 \in \mathscr{C}^\infty_0(\R)$ be a fixed cutoff function, equal to $1$ on a neighborhood of $0$.

\begin{assumption}[Coercivity]\label{hyp.subsolution}
There exists $c>0$ such that the following hold. For all $R \geq 1$ we can find $h_0= h_0(R) >0$ such that, for $h\in(0,h_0)$ and $(x_2,\xi_2)\in\R^2$, the function $\Phi$ satisfies
	\[ \Re \mathcal B(\xi_2+i\Phi'(x_2),x_2)+\Sigma(x_2)-b_0  \geq c R^2 h \Big( 1 - \chi_0 \left( \frac{x_2-c_u}{R h^{\frac 12}} \right) \Big) \,.\]
\end{assumption}
We emphasize that $\Phi$ may depend on $R$ and $h$. For simplicity we will use the notation $\chi_{R,h}(x) = \chi(x / R h^{\frac 12})$. We also recall that $\mathcal B$ has small variations, meaning that $|\mathcal B - b_0| \leq \varepsilon b_0$. The main result of this section is

\begin{proposition}\label{prop.functionalin}
Assume $\|\Sigma \|_\infty < (1-3\varepsilon)b_0$, with $\varepsilon$ from Assumption \ref{hyp.reg}. Let $C_0 >0$, and let $\Phi$ be a smooth bounded real-valued function with $\Phi' \in S(\R) + Rh^{\frac 12} S_{\frac 12}(\R)$, satisfying Assumption~\ref{hyp.subsolution}. Then there exists $C>0$ such that for $R$ large enough, we can find $h_0=h_0(R) >0$ such that
	\[ R^2 \| f\|_{L^2}\leq Ch^{-1}\|(\mathscr{M}_u^\Phi-z)f\|_{B^1 \otimes L^2}+C R^2 \|\chi_{R,h}(x_2-c_u)f\|_{L^2}+C\|f\|_{B^4 \otimes L^2}\,,\]
	 for all $h\in(0,h_0)$, all $f \in \mathcal{S}(\R^2)$, and all $z$ satisfying $|z-b_0| < C_0 h$.
\end{proposition}
Note that, when we say $\Phi' \in R h^{\frac 12} S_{\frac 12}(\R)$, we implicitly assume that the seminorms are uniformly bounded with respect to $R \in [1,\infty)$. This is the case for the weight we use later, see Lemma \ref{lem.phiuhat}. 
From Proposition \ref{prop.functionalin}, we deduce the following functional inequality, which avoids the complicated remainders related to the $B^N(\R)$-norm.

\begin{corollary}\label{cor.functionalin}
Let  $L\in\N$ and $C_0 >0$. If $\|\Sigma \|_\infty < (1-3\varepsilon)b_0$, there exists $N\in\N$ such that the following holds. 
Let $\Phi$ be a smooth, bounded, real-valued function with $\Phi' \in S(\R) + Rh^{\frac 12} S_{\frac 12}(\R)$ satisfying Assumption \ref{hyp.subsolution}. Then there exist $R, h_0, C>0$ such that
	\[\|f\|_{B^L \otimes L^2}\leq Ch^{-1}\|(\mathscr{M}_u^\Phi-z)f\|_{B^N \otimes L^2}+CR^2 \|\chi_{R,h}(x_2-c_u)f\|\,,\]
	for all $h\in(0,h_0)$, all $f \in \mathcal{S}(\R^2)$, and all $z$ satisfying $|z-b_0| < C_0 h$.
\end{corollary}

\begin{proof}[Proof of Corollary \ref{cor.functionalin}]
We start by establishing, for all $N \geq 0$,
\begin{equation}\label{eq.ellip1}
\| (\mathscr M_u^\Phi+1) f \|_{B^N \otimes L^2} \geq C \| f \|_{B^{N+2} \otimes L^2}. 
\end{equation}
Using the weight $W_N = (1+x_1^2 + \xi_1^2)^{\frac N2}$ and its quantization $W_N^w = \Op_1^{w,1} W_N$, it is enough to prove
\begin{equation}\label{eq.ellip2}
\| W^w_N (\mathscr M_u^\Phi+1) (W^{w}_{N+2})^{-1} g \|_{L^2} \geq C \| g \|_{L^2}. 
\end{equation}

By standard Weyl calculus, using that the symbol $m_u^\Phi$ behaves quadratically in $\xi_1$ and $x_1$, the symbol of the operator $ \mathscr N = W_N^w \mathscr M_u^\Phi (W^{w}_{N+2})^{-1}$ is in $S(\R)$, and the principal part of $\mathscr N^* \mathscr N$ is
\[ | W_N W_{N+2}^{-1} (m_u^\Phi +1) |^2 = W_{-4} \big| m_u^\Phi +1 \big|^2  \geq C.\]
The bound \eqref{eq.ellip2} follows thanks to the G\aa rding inequality for operator-valued symbols (see \cite[Theorem 2.1.18]{Keraval} in the case $\mathscr{A}=\mathscr{B}$), and thus \eqref{eq.ellip1} follows. Note that the involved constant $C$ can be chosen independent of $R$ and $\Phi$.

Now, we use Proposition~\ref{prop.functionalin} and \eqref{eq.ellip1} with $N \geq 2$, and we deduce
\begin{multline*}
 R^2 \| f \|_{L^2} \leq Ch^{-1} \|(\mathscr{M}_u^\Phi - z) f \|_{B^1 \otimes L^2} +CR^2 \|\chi_{R,h}(x_2-c_u)f\|_{L^2}\\ + C \| (\mathscr M_u^\Phi+1) f \|_{B^N \otimes L^2}, 
 \end{multline*}
  and since $z$ is bounded,
\[ R^2 \| f \|_{L^2} \leq Ch^{-1} \|(\mathscr{M}_u^\Phi - z) f \|_{B^N \otimes L^2}  +C R^2\|\chi_{R,h}(x_2-c_u)f\|_{L^2} + C \| f \|_{B^N \otimes L^2}. \]
Using \eqref{eq.ellip1} recursively, we can replace the $B^N$-norm of $f$ by the $L^2$ norm of $f$, up to increasing the constant $C$. Hence,
\[ R^2 \| f \|_{L^2} \leq C_N h^{-1} \|(\mathscr{M}_u^\Phi - z) f \|_{B^N \otimes L^2}  +CR^2\|\chi_{R,h}(x_2-c_u)f\|_{L^2} + C \| f \|_{L^2}.\]
We choose $R^2 >C$, and the result follows for $L=0$. For larger values of $L$, we use \eqref{eq.ellip1} recursively again.
\end{proof}

The rest of this section is devoted to the proof of Proposition \ref{prop.functionalin}.

\subsection{About the principal operator symbol}\label{sec.principalsymbol} We discuss basic properties of $M_0(X_2)$, the principal symbol of $\mathscr M_u^\Phi$ given in Proposition \ref{prop.symbol}.

\subsubsection{A harmonic oscillator with almost real coefficients}\label{sec.harmonic}
The operator $M_0(X_2)$ on $L^2({\mathbb R}_{x_1})$ is not selfadjoint. It has compact resolvent and its (discrete) spectrum can be computed explicitly. The spectrum is given by
\[\mathrm{sp}(M_0(X_2))=\big\{(2n-1)\mathcal B(\xi_2+i\Phi'(x_2),x_2) + \Sigma(x_2)\,,\quad n \in \mathbb N \big\}\,.\]
 These eigenvalues are algebraically simple, in the sense that the corresponding Riesz projector is of rank one.
 To see this, we can write
 \[\begin{split}M_0(X_2)&=(\mathcal B^2+ \alpha^2)D_{x_1}^2+\alpha (x_1 D_{x_1}+D_{x_1}x_1)+x_1^2 + \Sigma\\
 	&=(\mathcal B^2+\alpha^2)\left(D_{x_1}+\frac{\alpha}{\mathcal B^2+\alpha^2}x_1\right)^2+\frac{\mathcal B^2}{\mathcal B^2+\alpha^2} x_1^2 + \Sigma\,,
 \end{split}\]
 and we find a total family of eigenfunctions with the help of the Hermite functions. Note that here the functions $\mathcal{B}$ and $\alpha$ are taken at the point $(\xi_2 + i \Phi'(x_2),x_2)$,
even though we do not write it explicitly to lighten notations.
The normalized eigenfunction of $M_0$ associated with the eigenvalue $\mathcal B(\xi_2 + i \Phi'(x_2),x_2) + \Sigma(x_2)$ is
 \begin{equation}\label{def.FX2}
 \mathsf F_{X_2}(x_1)=C(X_2)\exp \Big( {-\frac{x_1^2}{2(\mathcal B-i \alpha)}} \Big) \,,
 \end{equation}
 where we choose the constant $C(X_2)>0$ such that $\|\mathsf F_{X_2}\|=1$. Note that this normalization is possible because ${\rm Re}\, (\mathcal B+i \alpha)$ stays positive, see Remark \ref{remark.positiverealpart}.

Taking the adjoint of $M_0$ amounts replacing $\mathcal B$ by $\overline{\mathcal B}$ and $\alpha$ by $\overline{\alpha}$. Thus, an eigenfunction of $M_0^*$ associated with $\overline{\mathcal B} + \Sigma$ is 
\begin{equation}\label{def.GX2}
\mathsf G_{X_2}(x_1)=c(X_2)\exp \Big( {-\frac{x_1^2}{2(\overline{\mathcal B}-i\overline{\alpha})}} \Big)\,,
\end{equation}
 where we choose the constant $c(X_2)\in\C\setminus\{0\}$ such that $\langle \mathsf F_{X_2}, \mathsf G_{X_2}\rangle=1$, for later convenience.

 \subsubsection{Inverting the principal operator symbol at fixed $X_2$}\label{sec.invertingX2}
\begin{lemma}\label{lem.Riesz}
The operator $\langle\cdot, \mathsf G_{X_2}\rangle \mathsf F_{X_2}$ is the Riesz projector $\Pi_{X_2}$ of $M_0(X_2)$ associated with the eigenvalue $\mathcal B ( \xi_2 + i \Phi'(x_2),x_2) + \Sigma(x_2)$. Moreover, the operator
\[R_0(z):=(M_0(X_2)-z)^{-1}(\mathrm{Id}-\Pi_{X_2})\,,\]
is a holomorphic function of $z$, for $|z-b_0| < b_0$.
\end{lemma}
 \begin{proof}
The Riesz projector is
 \[\Pi_{X_2}=\frac{1}{2i\pi}\int_{\mathscr C}(\zeta-M_0(X_2))^{-1}\mathrm{d}\zeta\]
 where $\mathscr C$ is a contour encircling the isolated eigenvalue $\mathcal B+ \Sigma$. Then, for all $z\neq \mathcal B + \Sigma$ inside $\mathscr C$ we have
 \[\begin{split}
 	\mathrm{Id}-\Pi_{X_2}&=\frac{1}{2i\pi}\int_{\mathscr C}\left((\zeta-z)^{-1}-(\zeta-M_0(X_2))^{-1}\right)\mathrm{d}\zeta\\
 	&=\frac{1}{2i\pi}\int_{\mathscr C}(\zeta-z)^{-1}(z - M_0(X_2))(\zeta-M_0(X_2))^{-1}\mathrm{d}\zeta
 	\end{split}\]
 so that
 \begin{equation}\label{eq.R0}
 (M_0(X_2)-z)^{-1}(\mathrm{Id}-\Pi_{X_2})=\frac{1}{2i\pi}\int_{\mathscr C}(\zeta-z)^{-1}(M_0(X_2)-\zeta)^{-1}\mathrm{d}\zeta
  \end{equation}
 is holomorphic inside $\mathscr C$, including at $z=\mathcal B + \Sigma$.
 
 Since $\langle \mathsf F_{X_2}, \mathsf G_{X_2}\rangle=1$, the operator $\langle\cdot, \mathsf G_{X_2}\rangle \mathsf F_{X_2}$ is a projector. It has the same range as $\Pi_{X_2}$, \emph{i.e.} $\mathrm{span}\, \mathsf F_{X_2}$. Moreover, it has the same kernel since $\ker \Pi_{X_2}=(\mathrm{ran}\, \Pi_{X_2}^*)^{\perp}=(\mathrm{span}\, \mathsf G_{X_2})^\perp$. Therefore $\Pi_{X_2} = \langle\cdot, \mathsf G_{X_2}\rangle \mathsf F_{X_2}$.
 
 Finally, note that we can choose $\mathscr C$ to be the circle of center $b_0$ and radius $b_0$, as soon as $\|\Sigma \|_{\infty} < (1-3\varepsilon)b_0$. Indeed, this path encircles $\mathcal B + \Sigma$ because
 \[ | \mathcal B + \Sigma -b_0 | \leq \varepsilon b_0 +|\Sigma| < b_0, \]
 and it does not encircles $(2n+1) \mathcal B + \Sigma$ for $n\geq 1$ since
\begin{align*}
 | (2n+1) \mathcal B + \Sigma - b_0 |&\geq 2n b_0 - |\Sigma| - (2n+1) | \mathcal B - b_0 | \\ &> (2n(1-\varepsilon) - \varepsilon - 1 + 3 \varepsilon) b_0 \geq b_0. 
 \end{align*}
\end{proof}
 
 The functions $\mathsf F_{X_2}$ and $\mathsf G_{X_2}$ are convenient for the construction of an augmented matrix of $M_0(X_2)-z$, which is bijective when $z$ is close to the first eigenvalue $\mathcal{B}(\xi_2+i\Phi'(x_2),x_2) + \Sigma(x_2)$ of $M_0(X_2)$. 
 \begin{proposition}\label{prop.P0z}
 	We let
 	\[\mathsf{P}_{0}(X_2)=\begin{pmatrix}
 		M_0(X_2)-z&\cdot \mathsf F_{X_2}\\
 		\langle\cdot, \mathsf G_{X_2}\rangle&0
 	\end{pmatrix} : B^2(\R)\oplus\C\to L^2(\R)\oplus \C\,.\]
If $|z-b_0| < b_0$, then $\mathsf{P}_{0}(X_2)$ is bijective and its inverse is given by
\begin{equation}\label{eq.P0-1}
\mathsf{P}_{0}(X_2)^{-1}=\begin{pmatrix}
R_0(z)&\cdot \mathsf F_{X_2}\\
\langle\cdot, \mathsf G_{X_2}\rangle&z- \mathcal B - \Sigma
\end{pmatrix}\,.
\end{equation}
Moreover, $\mathsf P_{0}(X_2)$ and $\mathsf P_{0}(X_2)^{-1}$ are bounded symbols.
 \end{proposition}
 \begin{proof}
 	One easily checks the formula for the inverse by taking the product of the two matrices. Clearly, $M_0(X_2)$, $\mathsf F_{X_2}$ and $\mathsf G_{X_2}$ are in the class of bounded symbols. It only remains to estimate $R_0(z)$ as function of $X_2$. $R_0(z)$ is given by formula \eqref{eq.R0}, where the contour $\mathscr C$ is independent of $X_2$. Therefore, it is enough to notice that $X_2 \mapsto (M_0(X_2) - \zeta)^{-1}$ is a smooth bounded function of $X_2$ with values in $\mathcal L(L^2(\R),B^2(\R))$, which depends continuously on $\zeta$. Similarly the $X_2$-derivatives are also bounded and therefore $\mathsf P_{0}(X_2)^{-1}$ belongs to the class
 	\[S(\R^2,\mathcal L( L^2(\R) \oplus \C, B^2(\R)\oplus \C)). \]
 	\end{proof}
 
 \subsection{Approximate parametrix}\label{sec.para}
We let $\mathscr{F}=\Op^w_h(\cdot\mathsf F_{X_2})$ and $\mathscr{G}=\Op^w_h(\langle \cdot , \mathsf G_{{X_2}}\rangle)$, and we consider
\[\mathscr{P}:=\begin{pmatrix}
	\mathscr{M}_u^\Phi-z&	\mathscr{F}\\
	\mathscr{G}	&0
\end{pmatrix}\,,\]
which is a pseudodifferential operator with principal symbol $\mathsf P_{0,z}(X_2)$. Using Proposition \ref{prop.P0z}, we will construct a parametrix for $\mathscr{P}$.  The existence of this explicit approximate inverse will give coercivity properties of $\mathscr{M}_u^\Phi$, eventually leading to the bound stated in Proposition \ref{prop.functionalin}. We define the subprincipal terms of $\mathscr P$ by 
\[\mathsf{P}_j(X_2)=\begin{pmatrix}
	M_j(X_2)&0	\\
	0&0
\end{pmatrix}\quad\mathrm{ and }\quad \mathscr{P}_j=\mathsf{P}_j^w\,,\]
where the subprincipal terms $M_j(X_2)$ of $\mathscr{M}_h^\Phi$ are defined in Proposition \ref{prop.symbol}.

\begin{lemma}\label{lem.expPh} For all $z \in \C$,
\begin{enumerate}
\item When $\Phi\in S(\R)$, for all $J\in\N$ we have the expansion
\[\mathscr{P}=\mathscr{P}_{0}+h^{\frac12}\mathscr{P}_1+h\mathscr{P}_2+\ldots+h^{\frac{J}{2}}\mathscr{P}_J+h^{\frac{J+1}{2}}\mathscr{R}_J\,,\]
where $\| \mathscr{R}_J \|_{( B^{J+3} \otimes L^2) \oplus L^2 \to (L^2 \otimes L^2) \oplus L^2}$ is uniformly bounded with respect to $h$.
\item When $\Phi'\in S(\R)+R h^{\frac12}S_{\frac12}(\R)$, we have
\[\mathscr{P}=\mathscr{P}_{0}+h^{\frac 12} \mathscr P_1 + h\mathscr{R}\,,\]
where $\| \mathscr{R} \|_{(B^4 \otimes L^2) \oplus L^2 \to (L^2 \otimes L^2) \oplus L^2}$ is uniformly bounded with respect to $h \in (0,1)$ and $R \in [1,\infty)$.
\end{enumerate}
	\end{lemma}

\begin{proof}
This result follows from  Proposition \ref{prop.symbol} and Lemma \ref{prop.symbolh12}, using the Calder\'on-Vaillancourt theorem (see \cite[Théorème 2.1.16]{Keraval}) to estimate the remainders. Note that, since the semi-norms of $\Phi'$ are uniformly bounded with respect to $R \in [1,\infty)$, the symbols involved belong to the class $S_{\frac 12, 0}(\R^2)$ uniformly with respect to $R$. Thus the estimates are independent of $R$. 
\end{proof}

\begin{proposition}\label{prop.parametrix}
	Assume $|z-b_0|<b_0$. Then there exists a sequence $(\mathsf{Q}_j)_{j\geq 0}$ of symbols in $S(\R^2,\mathcal L(B^j \oplus \C, B^2 \oplus \C))$ such that the following holds. Let $\mathscr{Q}_j=\Op_h^{w,2} \mathsf{Q}_j$, and for $J \in \N$, 
	\[\mathscr{Q}^{[J]} = \mathscr{Q}_0 + h^{\frac 12} \mathscr{Q}_1 + \ldots + h^{\frac J2} \mathscr{Q}_J. \]
\begin{enumerate}[\rm (i)]
	\item\label{eq.i} When $\Phi\in S(\R)$, we have for all $J \in \N$,
	\begin{equation}\label{eq.parametrixright}
	\mathscr{P}\mathscr{Q}^{[J]}=\mathrm{Id}+h^{\frac{J+1}{2}}\mathscr{R}_J\,,
	\end{equation}
	where $\mathscr R_J$ is bounded in $\mathcal L((B^{J+1} \otimes L^2) \oplus L^2,(L^2 \otimes L^2)\oplus L^2)$.
	\item When $\Phi'\in S(\R)+Rh^{\frac12}S_{\frac12}(\R)$, we have
	\begin{equation}\label{eq.parametrixleft}
	\mathscr{Q}^{[1]} \mathscr{P} =\mathrm{Id}+h\mathscr{R}\,,
	\end{equation}
	where $\mathscr{R}$ is bounded in $\mathcal L((B^4 \otimes L^2) \oplus L^2),(B^2 \otimes L^2)\oplus L^2)$.
	\end{enumerate}
	Moreover, we can choose $\mathsf{Q}_{0}=\mathsf{P}^{-1}_{0}$, $\mathsf{Q}_1=-\mathsf{Q}_0\mathsf{P}_1\mathsf{Q}_0$, and
		\begin{equation}\label{Q2eq}
		\mathsf{Q}_2=-\mathsf{Q}_0\mathsf{P}_2\mathsf{Q}_0+\mathsf{Q}_0\mathsf{P}_1\mathsf{Q}_0\mathsf{P}_1\mathsf{Q}_0-\frac{1}{2i}\{\mathsf{Q}_0,\mathsf{P}_0\}\mathsf{Q}_0\,,
		\end{equation}
		where we recall formula \eqref{eq.P0-1} for the inverse of $\mathsf P_0$.
	\end{proposition}
	
\begin{proof}
If $\mathsf{P}$ and $\mathsf{Q}$ denote the symbol of $\mathscr{P}$ and $\mathscr{Q}$ respectively, then the symbol of $\mathscr{Q}\mathscr{P}$ is given by the Moyal product (see the composition theorem in \cite[Theorem 2.1.12]{Keraval}),
	\begin{align*}
		 \mathsf Q \star \mathsf P &=  \mathsf Q \mathsf P + \frac{h}{2i}\lbrace \mathsf Q, \mathsf P \rbrace + \mathscr{O}(h^2)\\
		 &= \mathsf Q_0 \mathsf P_0 + h^{\frac 12} (\mathsf Q_0 \mathsf P_1 +\mathsf Q_1 \mathsf P_0)+ h \Big( \mathsf Q_2 \mathsf P_0 + \mathsf Q_1 \mathsf P_1 + \mathsf Q_0 \mathsf P_2 +  \frac{1}{2i} \lbrace \mathsf Q_0, \mathsf P_0 \rbrace \Big)  \\ &\qquad + \mathscr{O}(h^{\frac32})\,.
	\end{align*}
We obtain an expansion in powers of $h^{\frac 12}$. We want the first coefficient to equal the identity, and therefore $\mathsf Q_0 = \mathsf P_0^{-1}$. We want all other coefficients to vanish. This leads to a sequence of equations that can be solved recursively, to find the values of $\mathsf{Q}_j$. The first two equations give the claimed formulas for $\mathsf Q_1$ and $\mathsf Q_2$. Note that $\mathsf Q_j$ belongs to the symbol class $S(\R^2, \mathcal L(B^{j} \oplus \C,B^{2} \oplus \C))$. The final error is controlled by using Lemma \ref{lem.expPh}. The estimate of $\mathscr{P}\mathscr{Q}$ follows from the same arguments.
\end{proof}

\begin{remark}
As mentionned above, the expressions of the $\mathsf{Q}_j$ are the same as in \cite{MRVN23}. Note however that, in the context of \cite{MRVN23}, it is stated that $(\{\mathsf{Q}_0,\mathsf{P}_0\}\mathsf{Q}_0)_{\pm}=0$, which follows from an unfortunate mistake. To correct this, let us compute its explicit value in general. By writing that $\{\mathsf{Q}_0,\mathsf{P}_0\}\mathsf{Q}_0=\partial_{\xi_2}\mathsf{Q}_0\partial_{x_2}\mathsf{P}_0\mathsf{Q}_0-\partial_{x_2}\mathsf{Q}_0\partial_{\xi_2}\mathsf{P}_0\mathsf{Q}_0$ and by computing the product of matrices, we get
\[(\partial_{\xi_2}\mathsf{Q}_0\partial_{x_2}\mathsf{P}_0\mathsf{Q}_0)_\pm=\langle\partial_{x_2}\mathsf{P}_0 \mathsf F,\partial_{\xi_2} \mathsf G \rangle+\mathsf{Q}_0^\pm\langle\partial_{x_2} \mathsf F,\partial_{\xi_2} \mathsf G \rangle+\partial_{\xi_2}\mathsf{Q}^\pm_0\langle  \mathsf F,\partial_{x_2} \mathsf G\rangle\,.\]
Then, we notice that $(\mathsf{P}_0-E)\partial_{x_2}  \mathsf F=\partial_{x_2}E  \mathsf F-\partial_{x_2}\mathsf{P}_0  \mathsf F$. This implies that
\[(\partial_{\xi_2}\mathsf{Q}_0\partial_{x_2}\mathsf{P}_0\mathsf{Q}_0)_\pm=-\langle(\mathsf{P}_0-z)\partial_{x_2}  \mathsf F,\partial_{\xi_2} \mathsf G\rangle+\partial_{x_2}E\langle  \mathsf F,\partial_{\xi_2} \mathsf G \rangle-\partial_{\xi_2}E\langle  \mathsf F,\partial_{x_2}  \mathsf G \rangle\,.\]
By exchanging the roles of $x_2$ and $\xi_2$, we find that
\begin{multline*}
	(\{\mathsf{Q}_0,\mathsf{P}_0\}\mathsf{Q}_0)_{\pm}=-\langle(\mathsf{P}_0-z)\partial_{x_2}  \mathsf F,\partial_{\xi_2} \mathsf G\rangle+\langle(\mathsf{P}_0-z)\partial_{\xi_2}  \mathsf F,\partial_{x_2} \mathsf G\rangle\\
	+2\partial_{x_2}E\langle  \mathsf F,\partial_{\xi_2} \mathsf G\rangle-2\partial_{\xi_2}E\langle  \mathsf F,\partial_{x_2} \mathsf G\rangle\,.
	\end{multline*}
In \cite {MRVN23}, the correct expression of the effective symbol is
\begin{equation*}
	\begin{split}
	\mu^{\mathsf{eff}}_h(X_2) =& \rond B (X_2) + \rond{V}(X_2) \\
	&+ h \left( \langle \mathsf{P}_2(X_2) f_{X_2}, f_{X_2} \rangle - \langle \mathsf{P}_1(X_2) (\mathsf{P}_0(X_2) - z)^{-1} \Pi^{\perp} \mathsf{P}_1(X_2) f_{X_2}, f_{X_2} \rangle \right) \\
	&+\frac{h}{2i}	(\{\mathsf{Q}_0,\mathsf{P}_0\}\mathsf{Q}_0)_{\pm}\,.
\end{split}
\end{equation*}
\end{remark}

\subsection{Implications of the existence of the left inverse}

We deduce from Proposition \ref{prop.parametrix} a lower bound on $\mathscr{M}_u^\Phi$.

\begin{corollary}\label{cor.coercivity}
Let $C_0>0$ and consider $\Phi'\in S(\R)+Rh^{\frac12}S_{\frac12}(\R)$. Then there is $C>0$ such that, for $R\geq 1$, $h$ small enough, $z\in D(b_0,C_0h)$ and for all $f \in B^4(\R) \otimes L^2(\R)$,
	\[\|f\|_{L^2} \leq C\|\mathscr{G} f \|_{L^2}+C\|(\mathscr{M}_u^\Phi-z)f\|_{B^1 \otimes L^2}+Ch\|f \|_{B^4 \otimes L^2}\,,\]
	and
	\begin{multline*}
		\Re\langle \Op_h^{w,2} (\mathcal{B}(\xi_2 + i \Phi'(x_2),x_2) + \Sigma(x_2) -b_0) \mathscr{G}f,\mathscr{G}f\rangle \\ \leq C\|(\mathscr{M}_u^\Phi-z)f\|_{B^1 \otimes L^2} \|f\|_{L^2}+Ch\|f\|_{B^4 \otimes L^2} \| f \|_{L^2} \,.
		\end{multline*}
\end{corollary}

\begin{proof}
We use Proposition \ref{prop.parametrix}, and denote
\[ \mathscr Q^{[1]} =
\begin{pmatrix}
\mathscr Q_{++} & \mathscr Q_{+} \\ \mathscr Q_{-} & \mathscr Q_{\pm}
\end{pmatrix}.
 \]
Applying \eqref{eq.parametrixleft} to a vector $(f,0)$ gives
	\begin{equation}\label{eq.para1}
	\mathscr{Q}_{++}(\mathscr{M}_u^\Phi-z)f+\mathscr{Q}_{+}\mathscr{G}f=f+h\mathscr{R} f \,,
	\end{equation}
	and
	\begin{equation}\label{eq.para2}
	\mathscr{Q}_-(\mathscr{M}_u^\Phi-z)f+\mathscr{Q}_\pm \mathscr{G}f =h \tilde{\mathscr R}f \,,
	\end{equation}
	where the remainders satisfy 
	\[\| \mathscr R f\|_{B^2 \otimes L^2 }\leq C\|f\|_{B^4 \otimes L^2}\,,\quad \|\tilde {\mathscr R} f\|_{L^2} \leq C\|f\|_{B^4 \otimes L^2}\,.\]
	The first estimate follows from \eqref{eq.para1} since $\mathscr{Q}_{++}$ is bounded operator $B^1 \otimes L^2 \to L^2 \otimes L^2$ and $\mathscr{Q}_+$ is bounded $L^2 \to L^2 \otimes L^2$. For the second estimate, we take the scalar product with $\mathscr G f$ in \eqref{eq.para2} and obtain
	\begin{align*}
	 - {\mathrm{Re}} \langle \mathscr Q_{\pm} \mathscr G f, \mathscr G f \rangle &\leq \| \mathscr Q_- (\mathscr{M}_u^\Phi-z)f \|_{L^2} \|\mathscr G f \|_{L^2} + Ch \| f \|_{B^4 \otimes L^2} \| f \|_{L^2} \\
	 &\leq C \|(\mathscr{M}_u^\Phi-z)f \|_{B^1 \otimes L^2} \| f \|_{L^2}+ Ch \| f \|_{B^4 \otimes L^2} \| f \|_{L^2} ,
	 \end{align*}
	 where we used that $\mathscr Q_-$ is a bounded operator $B^1 \otimes L^2 \to L^2$. Using the formula from Proposition \ref{prop.parametrix} and equation \eqref{eq.P0-1}, we find the symbol of $\mathscr Q_{\pm}$,
	 \[ \mathsf Q_{\pm} = z - \mathcal B - \Sigma - h^{\frac 12} \langle M_1(X_2) \mathsf F_{X_2}, \mathsf G_{X_2} \rangle.\]
	 The subprincipal term is vanishing because $M_1(X_2)$ is odd in $x_1$ and $\mathsf F_{X_2}$, $\mathsf G_{X_2}$ are Gaussians. The result follows, replacing $z$ by $b_0$ up to additional $\mathscr{O}(h)$ error terms.
\end{proof}

\subsection{Proof of Proposition \ref{prop.functionalin}}\label{sec.proof}

The second inequality in Corollary \ref{cor.coercivity} implies that
	\begin{multline*}
	\Re\langle \Op_h^{w,2}( \mathcal{B}(\xi_2 + i \Phi'(x_2),x_2) + \Sigma(x_2) -b_0+ cR^2 h\chi_{R,h}) \mathscr{G} f,\mathscr{G}f\rangle \\ \leq C\|(\mathscr{M}_u^\Phi-z)f\|_{B^1 \otimes L^2} \|f\|_{L^2}+c R^2h\langle \chi_{R,h}\mathscr{G}f,\mathscr{G}f\rangle +Ch\|f\|_{B^4\otimes L^2} \| f \|_{L^2}\,,
	\end{multline*}
	where we introduced the function $\chi_{R,h}$ given in Assumption \ref{hyp.subsolution}, which satisfies
\[p_h(X_2):= \mathrm{Re}\, \mathcal{B}(\xi_2 + i \Phi'(x_2),x_2) + \Sigma(x_2) -b_0+cR^2h\chi_{R,h}(x_2)-cR^2h\geq 0\,.\]
In order to exploit the positivity, one will need a Fefferman-Phong inequality in a slightly exotic class.
\begin{lemma}[A Fefferman-Phong inequality]\label{lem.FP}
There exists $C>0$ such that, for all $R\geq 1$ and $h$ small enough,
\[\mathrm{Op}^{w,2}_h p_h\geq -Ch\,.\]
\end{lemma}
\begin{proof}
By translation and rescaling ($x_2=c_u+h^{\frac12}y$, $\xi_2=h^{-\frac12}\eta$), we see that $\mathrm{Op}^w_h p_h$ is unitarily equivalent to $\mathrm{Op}^{w}_1(q_h)$ with
\begin{multline*}
q_h(y,\eta)=\\ \mathrm{Re}\, \mathcal{B}(h^{\frac12}\eta + i \Phi'(c_u+h^\frac12 y),c_u+h^\frac12 y) + \Sigma(c_u+h^\frac12 y) -b_0+cR^2 h\chi_{0}(y)-cR^2 h\,.
\end{multline*}
By using that $\Phi' \in S(\R) + R h^{\frac 12} S_{\frac 12}(\R)$, we see that $\partial^\alpha (h^{-1}q_h)$ is uniformly bounded as soon as $|\alpha|\geq 4$. Therefore, we can apply the standard Fefferman-Phong inequality (see, for instance, \cite[Théorème 3.2]{Bony} or \cite[Corollary 1.3.2]{LM07}) and the conclusion follows.
\end{proof}
Thanks to Lemma \ref{lem.FP}, we deduce
\begin{multline*}
(c R^2-C)h\left\langle \mathscr{G}f,\mathscr{G}f\right\rangle\leq C\|(\mathscr{M}_u^\Phi-z)f\|_{B^1 \otimes L^2}\|f\|_{L^2}+cR^2 h\langle \chi_{R,h}\mathscr{G}f,\mathscr{G}f\rangle \\ +Ch\|f\|_{B^4\otimes L^2} \| f\| _{L^2}\,.
\end{multline*}
Then, we notice that
\[[\mathscr{G},\chi_{R,h}]=\mathscr{O}(h^{\frac12})\,,\]
which implies, up to reducing $h_0$,
\begin{multline*}
(cR^2-C)h\| \mathscr{G}f \|^2 \leq C\|(\mathscr{M}_u^\Phi-z)f\|_{B^1 \otimes L^2}\|f\|_{L^2}+CR^2 h\|\chi_{R,h}f\|_{L^2} \|f \|_{L^2}\\+Ch\|f\|_{B^4\otimes L^2} \| f \|_{L^2}\,.
\end{multline*}
By using the first estimate in Corollary \ref{cor.coercivity}, we get
\[(cR^2-C)h\|f\|_{L^2} \leq C\|(\mathscr{M}_u^\Phi-z)f\|_{B^1 \otimes L^2}+CR^2 h\|{\chi}_{R,h}f\|_{L^2} +Ch\|f\|_{B^4 \otimes L^2}\,.\]
This ends the proof of Proposition \ref{prop.functionalin}.

\section{Exponential decay}\label{sec.5}

\subsection{Magnetic Agmon distance}

In this section we prove an almost optimal result on the decay of eigenfunctions of $\mathscr{M}_u$, the single upper-well operator. More precisely, we prove decay in the variable $x_2$, which is the direction where tunneling occurs. The optimal decay rate is given by a function $\varphi_u(x_2)$, which we call \emph{Agmon distance} by analogy with electric wells. This function $\varphi_u$ is solution to an eikonal equation where $\mathcal{B}+ \Sigma$ is the effective symbol. This is where Assumption \ref{hyp.complexconnection} on the magnetic field is crucial.

\begin{lemma}\label{lem.Brond=b0}
Assume $\Sigma \in C^{\infty}_0(\R, \R_+)$ is such that $\Sigma(c_d)>0$ and $\| \Sigma \|_\infty$, $| {\rm{supp}}(\Sigma)|$ are small enough. Then there exists a smooth non-negative function $\varphi_u$ such that for all $x_2 \in \R$,
\begin{equation}\label{eq.eikonal}
 \mathcal B\left(i\varphi_u'(x_2),x_2\right) + \Sigma(x_2) =b_0\,. 
 \end{equation}
It satisfies $\varphi_u(c_u) = 0$, $\varphi_u'(c_u) = 0$ and $\varphi_u''(c_u)=-b_0 \gamma'(c_u)$. Moreover, outside the support of $\Sigma$ it is given by
\begin{equation}\label{eq.phius}
 \varphi_u(x_2) = \Big| \int_{c_u}^{x_2} |\Gamma(s)| \dd s \Big|, \qquad x_2 \in \R \setminus {\rm{supp}}( \Sigma ), \quad x_2 > c_d.
 \end{equation}
with
\[\Gamma(x_2):=\gamma(x_2)\int_{[0,1]} B(it \gamma(x_2),x_2)\mathrm{d}t\,.\]
	\end{lemma}
	\begin{proof}
		We recall that, for all $q\in\R^2$, 
	\[\mathcal B(\iota(q))=B(q)\,,\]
	where $\iota$ is the diffeomorphism defined in \eqref{def.iota}. Thus, by analyticity we have
	\[b_0 = B(i\gamma(q_2),q_2)=\mathcal B(\iota(i\gamma(q_2),q_2))=\mathcal B\left(\int_{[0,i\gamma(q_2)]} B(q_1,q_2)\mathrm{d}q_1 \,, q_2\right)\,.\]
	 By symmetry, this is also true replacing $\gamma$ by $-\gamma$, and we deduce that the smooth functions $\pm \Gamma$ satisfy
	 \begin{equation}
	 	\mathcal B( \pm i \Gamma(x_2), x_2) = b_0, \qquad \forall x_2 \in \R.
	 \end{equation}
	 We recall that $B(ix_1,x_2) >0$ for $x_1$, $x_2 \in \R$, and therefore $\Gamma$ has the same sign as $\gamma$,
	 \[ \begin{cases}
	 	\Gamma(x_2) < 0 \quad &{\rm{for}} \quad x_2 \notin [c_d,c_u],\\
	 	\Gamma(x_2) > 0 \quad &{\rm{for}} \quad x_2 \in (c_d,c_u).
	 	\end{cases}  \]
	For this reason the function
	\[ \psi(x_2) =  \Big| \int_{c_u}^{x_2} |\Gamma(s)| \dd s \Big|=
	\begin{cases}
	  \int_{x_2}^{c_u}\Gamma(s) \dd s &\quad {\rm{if}} \quad x_2 \geq c_d,\\
	  -\int_{x_2}^{c_d} \Gamma(s) \dd s + \int_{c_d}^{c_u}\Gamma(s) \dd s&\quad {\rm{if}} \quad x_2 \leq c_d,
	\end{cases}
	\]
	is non-negative, increasing on $[c_u,\infty)$, decreasing on $(- \infty, c_u]$, and its global minimum is $\psi(c_u)=0$. It is a solution to
	\[\mathcal B(i\psi'(x_2),x_2) = b_0,\]
	however it is not smooth at $c_d$. This is due to the double-well of $\mathcal B$ on the real axis, and it is fixed using the non-negative cutoff $\Sigma(x_2)$. We define $\varphi_u$ to be equal to $\psi$ outside the support of $\Sigma$. Equation \eqref{eq.eikonal} has locally a smooth solution, which cannot cross the axis $x_1=0$ because $\mathcal B(0,x_2) + \Sigma(x_2) > b_0$ away from $c_u$. Thus there must be a smooth function $\varphi_u'$ connecting the two branches of $\psi'$ on the support of $\Sigma$.
\end{proof}

\subsection{Weights satisfying Assumption \ref{hyp.subsolution}}\label{sec.weights}

The proof of exponential decay is based on the functional inequality Corollary \ref{cor.functionalin} derived in Section \ref{sec.functionaline}. However, we cannot apply this inequality directly to the weight $\varphi_u$ since it is not a bounded function. For this reason we introduce $\tilde \varphi_u$ which is constant at infinity. Let $A>c_u$, and consider $\chi_1 \in \mathscr{C}_0^{\infty}(\R)$ such that $\chi_1 = 1$ on $[-A,A]$, $\chi_1 = 0$ on $\R \setminus [-2A,2A]$ and $0 \leq \chi_1 \leq 1$. We define
\begin{equation}\label{eq:defphitilde}
\forall x \in \R, \quad {\tilde\varphi}_u(x) = \int_{c_u}^x \chi_1(s) \varphi_u'(s)\mathrm{d}s\,.
\end{equation}
One easily checks the following properties of $\tilde \varphi_u$.
	\begin{enumerate}[---]
		\item ${\tilde\varphi}_u=\varphi_u$ on $[-A,A]$,
		\item ${\tilde\varphi}_u$ is constant on $\R \setminus [-2A,2A]$,
		\item  ${\tilde\varphi}_u(\pm2A)<\varphi_u(\pm2A)$,
		\item for all $x\in\R$, $ (x-c_u)\tilde\varphi_u'(x) \geq 0$,
		\item $|{\tilde\varphi}_u' | \leq |\varphi_u'|$ and thus ${\tilde\varphi}_u \leq \varphi_u$.
	\end{enumerate} 
Moreover, the function $(1-\eta) \tilde \varphi_u \in S(\R)$ satisfies Assumption~\ref{hyp.subsolution}.

\begin{lemma} \label{lem.phiutilde}
If $\eta >0$ is small enough, then the function $(1-\eta) \tilde \varphi_u$ satisfies Assumption~\ref{hyp.subsolution}.
\end{lemma}

\begin{proof}
	The proof is divided into two steps.
	
	1. We first establish a couple of estimates for any weight $\Phi$. Thanks to the eikonal equation \eqref{eq.eikonal}, we can write
	\[\begin{split}
		&\Re \mathcal{B}(\xi_2+i\Phi',x_2) + \Sigma(x_2) -b_0\\
		&=\Re\left(\mathcal{B}(\xi_2+i\Phi',x_2)-\mathcal{B}(i\varphi_u',x_2)\right)\\
		&=\Re\left(\mathcal{B}(\xi_2+i\Phi',x_2)-\mathcal{B}(i\Phi',x_2)+\mathcal{B}(i\Phi',x_2)-\mathcal{B}(i\varphi_u',x_2)\right)\,.
	\end{split}\]
	The function $\xi_2\mapsto\Re\left(\mathcal{B}(\xi_2+i\Phi',x_2)-\mathcal{B}(i\Phi',x_2)\right)$ vanishes at $0$ and its derivative is the function
	$\xi_2\mapsto\Re\partial_1\mathcal{B}(\xi_2+i\Phi',x_2)=\Re\frac{\partial_{1}B}{B}(\iota^{-1}(\xi_2+i\Phi',x_2))$, which is non-negative when $\xi_2\geq 0$ and non-positive when $\xi_2\leq 0$ (as a consequence of Assumption \ref{hyp.odd}). Thus,
	\begin{equation}\label{eq.monotone}
		\Re\left(\mathcal{B}(\xi_2+i\Phi',x_2)-\mathcal{B}(i\Phi',x_2)\right)\geq 0\,.
	\end{equation}
	We deduce
	\begin{equation}
	\Re \mathcal{B}(\xi_2+i\Phi',x_2) + \Sigma(x_2) -b_0 \geq \mathcal{B}(i\Phi',x_2)-\mathcal{B}(i\varphi_u',x_2)\,,
	\end{equation}
	where the right-hand side is real due to the symmetry assumptions on $\mathcal B$. Then, we only have to bound the right-hand side from below. 
	
	 We use the Taylor formula to get
	\begin{equation}\label{eq.307}
	\mathcal{B}(i\Phi',x_2)-\mathcal{B}(i\varphi_u',x_2) = i\partial_1 \mathcal B(i\varphi_u',x_2)( \Phi' - \tilde\varphi_u') + \mathscr{O} \big( |\Phi' - \varphi_u'|^2 \big),
	\end{equation}
	and then
	\begin{multline}\label{eq.308}
	 \mathcal{B}(i\Phi',x_2)-\mathcal{B}(i\varphi_u',x_2) = - \partial_1^2 \mathcal B(0,x_2) \varphi_u' (\Phi' - \varphi_u') \\ + \mathscr{O} \big( (\varphi_u')^2 |\Phi' -\varphi_u'| +|\Phi' - \varphi_u'|^2\big)\,.
	\end{multline}
	Note that our assumptions on $B$ imply that, for all $x_2\in\R$, $\partial_1^2 \mathcal B(0,x_2) \geq c> 0$. Moreover,  thanks to Assumption \ref{hyp.complexconnection}, 
	\[ i\partial_1 \mathcal B(i\varphi'_u,x_2) \frac{\varphi_u'}{|\varphi_u'|} <0, \quad x_2 \neq c_u\,.\]

	2.  With $\Phi = (1-\eta) \varphi_u$, we deduce from \eqref{eq.308} that, 
	\begin{equation*} 
	\mathcal{B}(i(1-\eta)\varphi_u',x_2)-\mathcal{B}(i\varphi_u',x_2) \geq \partial_1^2 \mathcal B(0,x_2) \eta (\varphi_u')^2 + \mathscr{O} \big( \eta (\varphi_u')^3 + \eta^2 (\varphi_u')^2\big)\,.
	\end{equation*}
	We know that $\varphi_u$ is quadratic near $x_2=c_u$. Therefore, we can find a $c>0$ such that, if $\delta$ and $\eta$ are small enough,
	\begin{equation}\label{eq.lb1}
	\mathcal{B}(i(1-\eta)\varphi_u',x_2)-\mathcal{B}(i\varphi_u',x_2)  \geq c (x_2 - c_u)^2, \quad \forall x_2 \in D(c_u,\delta)\,.
	\end{equation}
	Then, we use \eqref{eq.307} so that, for all $x_2\in[-A,A]\setminus D(c_u,\delta)$,  we have
	\begin{align*}
	 \mathcal{B}(i(1-\eta)\varphi_u',x_2)-\mathcal{B}(i\varphi_u',x_2) &= i\partial_1 \mathcal B(i\varphi_u',x_2)( \Phi' - \varphi_u') + \mathscr{O} \big( |\Phi' - \varphi_u'|^2 \big) \\
	 & \geq c_\delta \eta |\varphi_u'|\,,
	 \end{align*}
	 for some $c_\delta >0$, and thus, for all $x_2\in[-A,A]\setminus D(c_u,\delta)$,
	 \begin{equation}\label{eq.lb2}
	  \mathcal{B}(i(1-\eta)\varphi_u',x_2)-\mathcal{B}(i\varphi_u',x_2)  \geq \tilde{c}_\delta >0\,.
	 \end{equation}
	 Since $\tilde\varphi_u=\varphi_u$ on $[-A,A]$, we get, from \eqref{eq.lb1} and \eqref{eq.lb2}, that for $\Phi=(1-\eta)\tilde\varphi_u$, on $[-A,A]$,
	 	 \[\mathcal{B}(i\Phi',x_2)-\mathcal{B}(i\varphi_u',x_2)\geq \tilde c_\delta\min(1,(x_2-c_u)^2)\,.\]
For all $x_2\notin[-A,A]$, we have
\[\mathcal{B}(i\Phi',x_2)-\mathcal{B}(i\varphi_u',x_2)\geq c_A>0\,.\]
Thus, $\Phi = (1-\eta) \tilde\varphi_u'$ satisfies the Assumption \ref{hyp.subsolution}. 
\end{proof}
	
In Proposition \ref{prop.approx} about WKB constructions, we will also need the following almost optimal weight.

\begin{lemma}\label{lem.phiuhat}
 Let $\rho$ be a smooth function equal to $1$ on $\R \setminus [-1,1]$ and supported in $\R \setminus [ - \frac 12, \frac 12 ]$. For all $R,\eta \geq 1$, define
 \[ \hat \varphi_u(x)=\tilde\varphi_u(x)- \eta R^2 h\int_{c_u}^x\rho\left(\frac{s-c_u}{Rh^{\frac12}}\right)\frac{\tilde{\varphi}'_u(s)}{\tilde{\varphi}_u(s)}\mathrm{d}s.\]
 Then if $\eta >0$ is small enough, $\hat \varphi_u$ satisfies Assumption \ref{hyp.subsolution}.
\end{lemma}

Note that $\hat \varphi_u'  \in S(\R) + Rh^{\frac 12} S_{\frac 12}(\R)$ uniformly with respect to $R$.

\begin{proof} 
We adapt the proof of Lemma \ref{lem.phiutilde}. We can use the bounds \eqref{eq.307} and \eqref{eq.308} to $\Phi = \hat \varphi_u$. Note that we have
		\[\hat \varphi_u'=\left(1-\frac{\eta R^2h\rho_h}{\tilde\varphi_u}\right)\tilde\varphi'_u\,,\]
		with $\rho_h(x) = \rho((x-c_u)/Rh^{\frac 12})$. Thus, \eqref{eq.308} gives, on $[-A,A]$,
		\begin{align*}
	&\mathcal{B}(i\hat \varphi_u',x_2)-\mathcal{B}(i\varphi_u',x_2) \\
	&\geq \partial_1^2 \mathcal B(0,x_2) \eta R^2 h \rho_h \frac{ (\varphi_u')^2}{ \varphi_u} + \mathscr{O} \Big(\eta R^2 h \rho_h \frac{ (\varphi_u')^3}{ \varphi_u} + \eta^2 R^4 h^2 \rho_h^2 \frac{ (\varphi_u')^2}{\varphi_u^2} \Big).
	\end{align*}
For $x_2 \in D(c_u,\delta)$ with $\delta$ small enough, using that $\varphi_u$ is quadratic near $c_u$, we deduce, 
	\begin{align*}
	 \mathcal{B}(i\hat \varphi_u',x_2)-\mathcal{B}(i\varphi_u',x_2) &\geq c \eta R^2 h \rho_h  + \mathscr{O} \Big(\eta R^2 h \rho_h |x_2-c_u| + \frac{\eta^2 R^2 h^2 \rho_h^2}{|x-c_u|^2} \Big) \\
	& \geq c\eta R^2 h \rho_h (1-C \delta - C \eta).
	\end{align*}
	When $\eta$ and $\delta$ are chosen small enough, we obtain, for all $x_2 \in D(c_u,\delta)$,
	\begin{equation}\label{eq.333}
	 \mathcal{B}(i\hat \varphi_u',x_2)-\mathcal{B}(i\varphi_u',x_2) \geq c\eta R^2 h \rho_h\,.
	\end{equation}
	 When $x_2 \in [-A,A]\setminus D(c_u,\delta)$, then $\rho_h(x_2) =1$ and we use \eqref{eq.307} to get
	\begin{align*}
	\mathcal{B}(i\hat \varphi_u',x_2)-\mathcal{B}(i\varphi_u',x_2) &\geq  c_\delta \eta R^2 h \frac{\varphi_u'}{\varphi_u} + \mathscr{O} \left( \eta^2 R^4 h^2 \frac{(\varphi_u')^2}{\varphi_u^2} \right) \\
	&\geq c_\delta \eta R^2 h - C_\delta \eta^2 R^4 h^2.
	\end{align*}
	We therefore obtain a $\tilde c_\delta>0$ (depending on $\delta$) such that, for all  $x_2 \in [-A,A]\setminus D(c_u,\delta)$,
	\begin{equation}\label{eq.334}
	\mathcal{B}(i\hat \varphi_u',x_2)-\mathcal{B}(i\varphi_u',x_2) \geq \tilde c_\delta \eta R^2 h\,.
	\end{equation}
	Moreover, we have, for all $x_2\notin[-A,A]$,
		\[\mathcal{B}(i\hat \varphi_u',x_2)-\mathcal{B}(i\varphi_u',x_2) \geq c_A>0\,.\]
	
	We combine \eqref{eq.333} and \eqref{eq.334} to deduce that $\hat{\varphi}_u$ satisfies Assumption \ref{hyp.subsolution}, as soon as $\eta$ is chosen small enough.
\end{proof}
	
	\begin{remark}\label{rem.AA}~
\begin{enumerate}[\rm (i)]
	\item The proofs of Lemmas \ref{lem.phiutilde} and \ref{lem.phiuhat} are rather similar to those of \cite[Prop. 5.3 \& 5.8]{AA23} obtained in the case of magnetic fields vanishing on curves. Note however that the proof of \cite[Prop. 5.8]{AA23} is not correct since the proof of the analog of \eqref{eq.monotone} (which follows from Assumption \ref{hyp.odd} in our case) is wrong. Fortunately, this  mistake can be corrected by improving the microlocalization estimates near "$\xi_2=0$".
		
		\item We have also to emphasize that our expressions of $\tilde\varphi_u$ and $\hat \varphi_u$ are taken from \cite{DR24} and are smooth functions contrary to those used in \cite{AA23, BHR22}. Due to the microlocal estimates, the smoothness of the weights is required.	
	\end {enumerate}
	\end{remark}

\subsection{Exponential decay of the one-well groundstate}\label{sec:One-well}

We are now in position to prove an almost-optimal decay estimate on eigenfunctions of $\mathscr{M}_u$.

\begin{theorem}\label{thm.agmon}
Let $\eta > 0$ small enough, $C_0>0$ and $L\in\N$. There exist $C,h_0>0$ such that, for all $h\in(0,h_0)$ and all eigenfunctions $\mathbf f$ associated with an eigenvalue $\lambda$ of $\mathscr{M}_u$ such that $|\lambda - b_0|\leq C_0h$, we have
\[\|e^{(1-\eta) \,\tilde\varphi_u(x_2)/h} \mathbf f\|_{B^L \otimes L^2}\leq C\|\mathbf f\|\,\]
where $\tilde\varphi_u$ is the weight defined in \eqref{eq:defphitilde}.
	\end{theorem}
	
\begin{proof}
We can assume $\eta$ as small as we need, and we apply Corollary \ref{cor.functionalin} to $\Phi=(1-\eta) \,\tilde\varphi_u$ (which is allowed thanks to Lemma \ref{lem.phiutilde}). Applying the inequality to $f= e^{(1-\eta) \,\tilde\varphi_u(x_2)/h}\mathbf f$ with $z=\lambda$, we obtain
\[ \| e^{(1-\eta) \,\tilde\varphi_u(x_2)/h} \mathbf f \|_{B^L \otimes L^2} \leq C \| \chi_{R,h}(x_2-c_u) e^{(1-\eta) \,\tilde\varphi_u(x_2)/h} \mathbf f \|, \]
for some $R>0$ and $C>0$. Since $\chi_{R,h}$ is supported on a $h^{\frac 12}$-neighborhood of $c_u$, and $\tilde{\varphi}_u$ is quadratic near $c_u$, the exponential on the right-hand side is bounded on the support of $\chi_{R,h}$. The result follows.
\end{proof}

\section{A WKB-like construction}\label{sec.6}

The parametrix constructed in Proposition \ref{prop.parametrix} can also be used to contruct quasimodes for the operator $\mathscr M_u^\Phi$, when $\Phi=\varphi_u$ is the \emph{magnetic Agmon distance} defined in Lemma \ref{lem.Brond=b0}. This yields a WKB-like construction for the upper-well operator $\mathscr M_u$.

\begin{theorem}\label{thm.WKB}
There exist sequences of smooth functions $(a_j)_{j \in \N}$ and of real numbers $(z_j)_{j \in \N}$, and $N>0$ such that the following holds for all $J \in \N$. Define
\[a^{[J]}=a_0+h^{\frac12}a_1+\ldots+h^{\frac{J}{2}}a_J\,,\quad z^{[J]}=z_0+h^{\frac12}z_1+\ldots+h^{\frac{J}{2}}z_J\,,\]
and let $\chi \in \mathscr{C}^\infty_0(\R)$ be such that $\chi=1$ on a neighborhood of $c_u$.
\begin{enumerate}
\item There exists an $h$-pseudodifferential operator $\mathscr{Q}_+^{[J]}$ with principal symbol
\[ \mathsf{Q}^{+} = \mathsf{F}_{(x_2,\xi_2)}(x_1) + \mathscr{O}( h ^{\frac 12}) \quad {\rm in} \quad S(1), \]
such that
\[(\mathscr{M}_u^{ \varphi_u}-z^{[J]})\mathscr{Q}^{[J]}_+[\chi a^{[J]} ]=\mathscr{O}(h^{\frac{J+1}{2}}) \quad {\rm{in}} \quad B^L(\R) \otimes L^2(\lbrace \chi =1 \rbrace)\,,\]
for all $L \in \N$.
\item The quasimode
\[\mathbf{f}^{\rm{wkb}}_{J,u}(x_1,x_2)=\varphi_u''(c_u)^{\frac 14} h^{-\frac14}e^{-\varphi_u(x_2)/h}\mathscr{Q}^{[J]}_+ \left[\chi a^{[J]} \right](x_1,x_2)\]
satisfies
\[e^{\varphi_u/h}(\mathbf{f}_{J,u}^{\rm{wkb}}-\Pi_u \mathbf{f}_{J,u}^{\rm{wkb}})=\mathscr{O}(h^{\frac{J}{2}-1-N }) \quad \textrm{in} \quad B^L(\R) \otimes L^2(\lbrace \chi =1 \rbrace)\,,\]
for all $L \in \N$, and where $\Pi_u$ is the orthogonal projection on the ground state of $\mathscr{M}_u$.
\end{enumerate}
Moreover $z_0 = b_0$, $ z_1 = 0$, $ a_0(c_u) = 1$, and $\|\mathbf{f}^{\rm{wkb}}_{J,u}\| = 1 + \mathscr{O}(h^{\frac 12})$.
\end{theorem}

This results tells us that $\mathbf{f}^{\rm{wkb}}_{J,u}$ is a good approximation of the true ground state $\mathbf f_u$ of $\mathscr{M}_u$. This section is devoted to the proof of Theorem \ref{thm.WKB}.

\subsection{A sequence of transport equations}

We want to use the parametrix construction in Proposition \ref{prop.parametrix} with weight $\varphi_u$. Since $\varphi_u$ is not bounded, we use instead $\Phi = \tilde{\varphi}_u$ defined in \eqref{eq:defphitilde}. We may choose the constant $A>0$ defining $\tilde{\varphi}_u$ such that ${\rm{supp}}\, \chi \subset [-A,A]$. With this choice, and for any $z$ of the form 
\[z = b_0 + z_1 h^{\frac 12} + \cdots + z_J h^{\frac{J}{2}},\]
we obtain an $h$-pseudodifferential operator
\[ \mathscr{Q}^{[J]} = 
\begin{pmatrix}
\mathscr Q_{++}^{[J]} & \mathscr Q_{+}^{[J]} \\ \mathscr Q_{-}^{[J]} & \mathscr Q_{\pm}^{[J]}
\end{pmatrix},
 \]
such that
\begin{equation}\label{eq.param2}
 \begin{pmatrix}
\mathscr{M}_u^{\tilde{\varphi}_u} - z & \mathscr{F} \\ \mathscr{G} & 0
\end{pmatrix}
 \mathscr{Q}^{[J]} = I + \mathscr{O}(h^{\frac{J+1}{2}}).
\end{equation}
We recall that $\mathscr{F} = \Op_h^{w} \mathsf F_{(x_2,\xi_2)}(x_1)$, with the notation \eqref{def.FX2}. Proposition \ref{prop.parametrix} also tells us that $\mathscr{Q}_+^{[J]} = \mathscr{F} + \mathscr{O}(h^{\frac 12})$ in the symbol class $S(1)$. Applying \eqref{eq.param2} to a function of the form $(0,\chi (x_2) a(x_2))$, we deduce that
\begin{equation}\label{eq.inv1}
(\mathscr{M}_u^{\tilde{\varphi}_u}-z)\mathscr{Q}^{[J]}_+(\chi a)=-\mathscr{F}\mathscr{Q}^{[J]}_\pm (\chi a)+ \mathscr{O}\Big( h^{\frac{J+1}{2}} \| \chi a \| \Big) \,,
\end{equation}
in $L^2(\R^2)$, and
\begin{equation}\label{eq.inv2}
\mathscr{G} \mathscr{Q}^{[J]}_+(\chi a )=\chi a + \mathscr{O} \Big( h^{\frac{J+1}{2}} \| \chi a \| \Big)\,,
\end{equation}
in $L^2(\R)$. Our quasimode $\mathbf{f}^{\rm{wkb}}_{J,u}$ should solve the eigenvalue equation for $\mathscr{M}_u$, which suggests to find $z$ and $a$ such that $\mathscr{Q}_{\pm}^{[J]}(\chi a)$ is small (in \eqref{eq.inv1}). This is the purpose of the following lemma.

\begin{lemma}\label{lem.transport}
There exist sequences of smooth functions $(a_j)_{j \in \mathbb N}$ and of numbers $(z_j)_{j \in \mathbb N}$ such that
\[ \mathscr{Q}_\pm^{[J]} ( \chi a^{[J]} ) = \mathscr{O}(h^{\frac{J+1}{2}}) \quad {\rm{in}} \quad L^2(\lbrace \chi =1 \rbrace), \]
for all $J \in \N$, where $a^{[J]} = a_0 + h^{\frac 12}a_1 + \cdots + h^{\frac J2} a_J$ and $z = z_0 +h^{\frac 12} z_1 + \cdots h^{\frac J2} z_J$.
Moreover, 
\[z_0 = b_0, \quad z_1 = 0, \quad z_2 = \frac{\sqrt{\partial_1^2 B(C_u) \partial_2^2 B(C_u)}}{2b_0} + \frac{\big( \sqrt{\partial_1^2 B(C_u)} + \sqrt{\partial_2^2B(C_u)} \big)^2}{4b_0}, \]
and $a_0$ is explicitly given in equation \eqref{def.a0} below.
\end{lemma}

\begin{proof}
We recall that $\mathscr{Q}^{[J]}_\pm$ is a pseudodifferential operator with symbol in $S(\R^2)$ of the form
\begin{equation}
\mathsf{Q}^\pm = \mathsf Q_0^{\pm} + h^{\frac 12} \mathsf Q^{\pm}_1 + h \mathsf Q^{\pm}_2 + \cdots + h^{\frac{J}{2}} \mathsf Q^{\pm}_J.
\end{equation}
with 
\begin{equation}
\begin{cases}
\mathsf Q_0^{\pm}(x_2,\xi_2) = z - \mathcal B(\xi_2 + i \tilde{\varphi}_u'(x_2),x_2) - \Sigma(x_2) \\
\mathsf Q_1^{\pm}(x_2,\xi_2) = 0.
\end{cases}
\end{equation}
Now note that, given any symbol $\sigma \in S(1)$ and any function $a \in \mathscr{C}^\infty_0(\R)$, we have 
\begin{equation}\label{exp.Ops}
\Op_h^w ( \sigma(x,\xi) ) a = \Op_1^w( \sigma(x,h\xi)) a = \sum_{n=0}^{N} \frac{h^n}{n!} \Op_1^w ( \partial_\xi^n \sigma(x,0) \xi^n) a + \mathscr{O}(h^{N+1}).
\end{equation}
Therefore, we can expand our quantity of interest
\begin{equation}\label{exp.Q}
 \mathscr{Q}_\pm^{[J]}( \chi a_h) = \Big( \Op_h^w ( \mathsf Q_0^\pm) + h^{\frac 12} \Op_h^w (\mathsf Q_1^\pm) + \cdots \Big) (\chi a_0 + h^{\frac 12} \chi a_1 + \cdots ) 
 \end{equation}
in powers of $h^{\frac 12}$. We want every term of the expansion to vanish, and this gives rise to a series of equation. We explain below how to find the sequence $(a_j)_{j\in \mathbb N}$ that solves these equations, and the claimed result will follow. Note that, since the terms in \eqref{exp.Ops} are differential operators, and since we only need to estimate \eqref{exp.Ops} on $\lbrace \chi  = 1 \rbrace$, we can remove $\chi$ in the equations.

\medskip
\textit{Equation of order} $h^0$. The first equation arising from \eqref{exp.Q} is
\[ \mathcal B(i \tilde{\varphi}_u'(x_2),x_2) + \Sigma(x_2) = z_0 .\]
Since we want $z$ to be close to $b_0$, we choose $z_0=b_0$. On $\lbrace \chi =1 \rbrace$, we have $\tilde{\varphi}_u = \varphi_u$ and the equation is solved by Lemma \ref{lem.Brond=b0}. We then have $\mathsf{Q}_0^\pm(x_2,0) = z-z_0$.

\medskip
\textit{Equation of order} $h^{\frac 12}$. Due to the first equation, and since $\mathsf{Q}_1^{\pm} = 0$, it only remains $z_1 a_0 = 0$, which is solved by taking $z_1=0$.

\medskip
\textit{Equation of order} $h$. The terms of order $h$ in equation \eqref{exp.Q} give a differential equation for $a_0$,
\begin{equation}\label{transport0}
 \mathcal T a_0 = 0, 
 \end{equation}
with
\[ \mathcal T = i \partial_{1} \mathcal B( i \varphi_u'(x_2),x_2) \frac{\partial}{\partial x_2} + \frac 12 \frac{d}{dx_2} \Big( i \partial_{1} \mathcal B(i \varphi_u'(x_2),x_2) \Big) + \mathsf Q_2^{\pm}(x_2,0) + z_2. \]
In comparison to standard WKB constructions, $\mathcal T$ is our \emph{transport} equation.
Note that, due to our assumptions on $\mathcal B$, the coefficient $i \partial_{1} \mathcal B(i \varphi_u'(x_2),x_2)$ is real and vanishes only when $x_2=c_u$, and to first order only. If we choose

\[ z_2 = \frac{\varphi_u''(c_u)}{2}\partial_{1}^2 \mathcal B(0,c_u) - \mathsf{Q}_2^{\pm}(c_u,0), \]
then the function \begin{equation}\label{eq.D} D(x_2) =  \frac 12 \frac{d}{dx_2} \Big( i \partial_{1} \mathcal B(i \varphi_u'(x_2),x_2) \Big) + \mathsf Q_2^{\pm}(x_2,0) + z_2\end{equation} vanishes to first order at $x_2=c_u$. Thus, equation \eqref{transport0} has a smooth solution,
\begin{equation}\label{def.a0}
a_0(x_2) = \exp \Big( - \int_{c_u}^{x_2} \frac{D(s)}{i \partial_{1} \mathcal B(i \varphi_u'(s),s)} \dd s \Big).
\end{equation}

\textit{Equation of order} $h^{\frac{j+2}{2}}$ \textit{for} $j \geq 1$. We then follow the usual strategy, with details carried over in \cite{DS99} for instance. For all $j \geq 1$, we obtain an equation for $a_j$ of the form
\begin{equation}\label{eq.transport2}
\mathcal T a_j + z_{j+2} a_0 = g_j(x_2; a_0, \cdots, a_{j-1}, z_0, \cdots z_{j+1}),
\end{equation}
where $g_j$ is some function of the quantities fixed at previous steps. We can always choose $z_{j+2}$ such that \eqref{eq.transport2} has as a smooth solution $a_j$.
\end{proof}

\subsection{Properties of the WKB-like state}

Lemma \ref{lem.transport} provides us with a sequence $(a_j)_{j\in \N}$ defining the quasimode $\mathbf{f}^{\rm{wkb}}_{J,u}$. It then follows from \eqref{eq.inv1} that
\[(\mathscr{M}_u^{\tilde \varphi_u}-z^{[J]})\mathscr{Q}^{[J]}_+[\chi a^{[J]} ]=\mathscr{O}(h^{\frac{J+1}{2}}) \quad {\rm{in}} \quad B^N(\R) \otimes L^2(\lbrace \chi =1 \rbrace)\,,\]
for all $N \in \N$. Since $\tilde \varphi_u = \varphi_u$ on ${\rm{supp}} \, \chi$, and by definition of $\mathbf{f}^{\rm{wkb}}_{J,u}$ we deduce
	\begin{equation} \label{eq.bkwpoid}
	e^{\varphi_u/h}(\mathscr{M}_u-z^{[J]}) \mathbf{f}^{\rm{wkb}}_{J,u} =\mathscr{O}(h^{\frac{J}{2} +\frac 14} ) \quad {\rm{in}} \quad B^N(\R) \otimes L^2(\lbrace \chi =1 \rbrace)\,.
	\end{equation}
We also get

\begin{lemma}\label{lem.wkb}
The state $\mathbf{f}^{\rm{wkb}}_{J,u}$, and the real number $z^{[J]}$ constructed using Lemma~\ref{lem.transport} are such that:
\begin{enumerate}
\item For all $N \in \N$, 
	\[ \| (\mathscr{M}_u - z^{[J]}) \mathbf f^{\rm{wkb}}_{J,u} \|_{B^N \otimes L^2} = \mathscr{O}\big(h^{\frac{J}{2} + \frac 14}\big). \]
\item The ground state energy $\mu$ of $\mathscr{M}_u$ satisfies
\[|z^{[J]} - \mu | = \mathscr{O}\big(h^{\frac{J}{2}+ \frac 14}\big). \]
\item The WKB state is asymptotically normalized: $\|\mathbf{f}^{\rm{wkb}}_{J,u}\| = 1 + \mathscr{O}(h^{\frac 12})$.
\end{enumerate}
\end{lemma}

\begin{proof}
The first statement follows from \eqref{eq.bkwpoid} using $\varphi_u \geq 0$ on $\lbrace \chi =1 \rbrace$, and $\varphi_u  >0$ on $\overline{\lbrace \chi \neq 1 \rbrace}$. Then, by the spectral theorem we deduce that
\[ |\mu - z^{[J]} | = \mathscr{O}\big(h^{\frac{J}{2}+\frac 14}  \big), \]
since $\mu$ is the only eigenvalue of $\mathscr{M}_u$ behaving like  $b_0+hz_2+\mathcal{O}(h^{\frac32})$.
Finally, we can estimate the norm of $\mathbf{f}^{\rm{wkb}}_{J,u}$ using that
\[ \mathscr{Q}^{[J]}_+ = \mathscr{F} + \mathscr{O}(h^{\frac12}) = \Op_h^{w,2}\big( \mathsf F_{(x_2,\xi_2)}(x_1) \big) + \mathscr{O}(h^{\frac 12}). \]
Indeed, we then have
\begin{align*}
\| \mathbf f^{\rm{wkb}}_{J,u} \|^2 &= \sqrt{\frac{\varphi_u''(c_u)}{h}} \| e^{- \varphi_u / h } \Op_h^{w,2} \big( \mathsf F_{(x_2,\xi_2)} \big) \chi a_0 \|^2 ( 1 + \mathscr{O}(h^{\frac 12}))\\
&=  \sqrt{\frac{\varphi_u''(c_u)}{h}} \int e^{-2 \varphi_u(x_2) /h} |\mathsf{F}_{(x_2,0)}(x_1) \chi(x_2) a_0(x_2)|^2 \dd x_1  \dd x_2 \big( 1 + \mathscr{O}(h^{\frac 12}) \big)\\
&= \sqrt{\frac{\varphi_u''(c_u)}{h}} \int e^{-2 \varphi_u(x_2) /h} |\chi(x_2) a_0(x_2)|^2 \dd x_2 \big( 1+ \mathscr{O}(h^{\frac 12}) \big),
\end{align*}
because $x_1 \mapsto \mathsf F_{X_2}(x_1)$ has norm $1$. We estimate this integral using the stationary phase method. We deduce $\| \mathbf f^{\rm{wkb}}_{J,u} \| = 1 + \mathscr{O}(h^{\frac 12})$ since $\chi(c_u)a_0(c_u)=1$.
\end{proof}

\subsection{Proof of Theorem \ref{thm.WKB}}

Using the construction above, Theorem \ref{thm.WKB} is then a consequence of the following lemma.

\begin{lemma}\label{prop.approx}
There exists $N >0$ such that, for all $L, J \in \N$ we have
\[e^{\varphi_u/h}(\mathbf{f}_{J,u}^{\rm{wkb}}-\Pi_u \mathbf{f}_{J,u}^{\rm{wkb}})=\mathscr{O}(h^{\frac{J}{2}-1-N }) \quad \textrm{in} \quad B^L(\R) \otimes L^2(\lbrace \chi =1 \rbrace)\,.\]
\end{lemma}

\begin{proof}
We first notice that
\begin{equation}\label{eq.norm.app}
\| \mathbf f^{\rm{wkb}}_{J,u}- \Pi_u  \mathbf f^{\rm{wkb}}_{J,u} \|_{L^2(\R^2)} = \mathscr{O} \big( h^{\frac{J}{2}-\frac 34} \big)\,,
\end{equation}
which follows from
\begin{align*}
h b_0 \| \mathbf f^{\rm{wkb}}_{J,u}- \Pi_u  \mathbf f^{\rm{wkb}}_{J,u} \|^2 &\leq \langle \mathscr M_u (\mathbf f^{\rm{wkb}}_{J,u}- \Pi_u  \mathbf f^{\rm{wkb}}_{J,u}), \mathbf f^{\rm{wkb}}_{J,u}- \Pi_u  \mathbf f^{\rm{wkb}}_{J,u} \rangle \\
&\leq \| \mathscr{M}_u \mathbf f^{\rm{wkb}}_{J,u} - \mu  \Pi_u  \mathbf f^{\rm{wkb}}_{J,u} \| \| \mathbf f^{\rm{wkb}}_{J,u}- \Pi_u  \mathbf f^{\rm{wkb}}_{J,u} \|\\
& \leq \big( \| (\mathscr{M}_u - z^{[J]}) \mathbf f^{\rm{wkb}}_{J,u} \|+ |z^{[J]} - \mu | \|  \mathbf f^{\rm{wkb}}_{J,u} \| \big)  \| \mathbf f^{\rm{wkb}}_{J,u}- \Pi_u  \mathbf f^{\rm{wkb}}_{J,u} \|,
\end{align*}
and Lemma \ref{lem.wkb}.

We then use Corollary \ref{cor.functionalin} with $z=\mu$, $f= e^{\hat{\varphi}_u/h} \big( \mathbf f^{\rm{wkb}}_{J,u}- \Pi_u  \mathbf f^{\rm{wkb}}_{J,u} \big)$, and the weight $\Phi=\hat{\varphi}_u$ defined in Lemma \ref{lem.phiuhat}. We obtain
\begin{multline}
\| e^{\hat{\varphi}_u/h} \big( \mathbf f^{\rm{wkb}}_{J,u}- \Pi_u  \mathbf f^{\rm{wkb}}_{J,u} \big) \|_{B^L \otimes L^2} \leq \frac{C}{h} \|e^{\hat{\varphi}_u/h} (\mathscr{M}_u - \mu) ( \mathbf f^{\rm{wkb}}_{J,u}- \Pi_u  \mathbf f^{\rm{wkb}}_{J,u}) \|_{B^N \otimes L^2} \\ 
+ C \|\chi_{R,h} e^{\hat{\varphi}_u/h} (\mathbf f^{\rm{wkb}}_{J,u}- \Pi_u  \mathbf f^{\rm{wkb}}_{J,u}) \|\,,
\end{multline}
for some $C, R >0$.
On the support of $\chi_{R,h}$, the function $\hat{\varphi}_u$ is of order $h$. Therefore, the last term can be estimated by using \eqref{eq.norm.app}. We deduce that
\begin{multline}
\| e^{\hat{\varphi}_u/h} \big( \mathbf f^{\rm{wkb}}_{J,u}- \Pi_u  \mathbf f^{\rm{wkb}}_{J,u} \big) \|_{B^L \otimes L^2} \leq \frac{C}{h} \|e^{\hat{\varphi}_u/h} (\mathscr{M}_u - \mu) ( \mathbf f^{\rm{wkb}}_{J,u}- \Pi_u  \mathbf f^{\rm{wkb}}_{J,u}) \|_{B^N \otimes L^2} \\ + \mathscr{O}\big(h^{\frac{J}{2}-\frac 34} \big).
\end{multline}
Since $\Pi_u$ projects on the kernel of $\mathscr{M}_u -\mu$, we get
\begin{equation*}
\| e^{\hat{\varphi}_u/h} \big( \mathbf f^{\rm{wkb}}_{J,u}- \Pi_u  \mathbf f^{\rm{wkb}}_{J,u} \big) \|_{B^L \otimes L^2} \leq \frac{C}{h} \|e^{\hat{\varphi}_u/h} (\mathscr{M}_u - \mu) \mathbf f^{\rm{wkb}}_{J,u} \|_{B^N \otimes L^2} + \mathscr{O}\big(h^{\frac{J}{2}-\frac 34} \big).
\end{equation*}
We insert the weight $e^{\varphi_u/h}$ and replace $\mathbf{f}^{\rm{wkb}}_{J,u}$ by its expression to get
\begin{multline*}
\| e^{\hat{\varphi}_u/h} \big( \mathbf f^{\rm{wkb}}_{J,u}- \Pi_u  \mathbf f^{\rm{wkb}}_{J,u} \big) \|_{B^L \otimes L^2} \leq \frac{C}{h^{5/4}} \|e^{(\hat{\varphi}_u - \varphi_u)/h} (\mathscr{M}_u^{\varphi_u} - \mu) \mathscr{Q}_+^{[J]}(\chi a^{[J]}) \|_{B^N \otimes L^2} \\+ \mathscr{O}\big(h^{\frac{J}{2}-\frac 34}\big)\,.
\end{multline*}
Since $\hat \varphi_u \leq \varphi_u$, we can use Lemma \ref{lem.wkb} to find
\[ \| e^{\hat{\varphi}_u/h} \big( \mathbf f^{\rm{wkb}}_{J,u}- \Pi_u  \mathbf f^{\rm{wkb}}_{J,u} \big) \|_{B^L \otimes L^2} = \mathscr{O}\big(h^{\frac{J}{2} - 1} \big). \]
Finally, we use that $\hat \varphi_u \geq \tilde \varphi_u - CR^2 h|\ln h|$. This gives an $N>0$ such that
\[ h^N \| e^{\tilde{\varphi}_u/h} \big( \mathbf f^{\rm{wkb}}_{J,u}- \Pi_u  \mathbf f^{\rm{wkb}}_{J,u} \big) \|_{B^L \otimes L^2} = \mathscr{O}\big(h^{\frac{J}{2} - 1} \big),\]
and the result follows because $\tilde \varphi_u = \varphi_u$ on $\lbrace \chi =1 \rbrace$.
\end{proof}

\section{Interaction matrix and proof of Theorem~\ref{thm.main} }\label{sec.7}

\subsection{Interaction matrix}

Once we have the exponential decay of the eigenfunctions given by Theorem \ref{thm.agmon}, we can adapt the Helffer-Sjöstrand theory \cite{HS84} to obtain an effective interaction matrix describing the spectrum of the double-well operator $\mathscr{M}$. 
Recall that by Proposition~\ref{prop.Mh} the double-well magnetic Laplacian $\mathscr{L}$ is unitarily equivalent to $h \mathscr{M}$, so it suffices to understand the eigenvalues of $\mathscr{M}$.
We consider the eigenspace \[E=\ker(\mathscr{M}-\lambda_1)\oplus\ker(\mathscr{M}-\lambda_2)\] and let $\Pi_E$ be the associated orthogonal projection. We construct a basis for $E$ using the ground states $\mathbf f_u$ and $\mathbf f_d$ of $\mathscr{M}_u$ and $\mathscr{M}_d$ respectively. Recall that the \emph{up} and \emph{down} operators 
\[ \mathscr{M}_u = \mathscr{M} + \Sigma(x_2), \qquad \mathscr{M}_d = \mathscr{M} + \Sigma(-x_2) \]
are related by the symmetry $U \mathbf f (x) = \mathbf f (-x)$:
\[ U \mathscr{M}_u U = \mathscr{M}_d, \qquad \mathbf f_d = U \mathbf f_u.\]
We then define the truncated versions of $\mathbf f_u$ and $\mathbf f_d$,
\[f_{{u}}=\chi_{{u}}\mathbf f_{{u}}\,,\qquad f_{{d}}=\chi_{{d}}\mathbf f_{{d}} = U f_u\,,\]
where $\chi_{{d}}(x_2)=U\chi_{{u}}(x_2)=\chi_{{u}}(-x_2)$ and the cutoff $\chi_u$ is such that
\begin{equation}\label{def.chiu}
\chi_u(x_2) = 
\begin{cases}
0 &{\textrm{on }}  {\rm{supp}}( \Sigma ),\\
1 &{\textrm{away from a small neighborhood of }} {\rm{supp}}( \Sigma ).
\end{cases}
\end{equation}
Using the decay of $\mathbf f_u$ and $\mathbf f_d$, we deduce that they define good approximate eigenfunctions for $\mathscr{M}$. This is a corollary of Theorem \ref{thm.agmon}.

\begin{proposition}\label{prop.quasimodesexp}
Let $\kappa >0$. Then, if $\|\Sigma \|_{\infty}$ and $|{\rm{supp}}( \Sigma )|$ are small enough, for $\star\in\{u,d\}$, we have
\[\|(\mathscr{M}-\mu) f_\star\|=\mathscr{O}(e^{-(1- \kappa) \frac{S}{h}})\|f_\star\|\,,\]
with $S= \int_{c_d}^{c_u} \Gamma(s) \dd s$, and where $\mu$ is the ground state energy of $\mathscr{M}_u$.
Moreover, the first eigenvalues of $\mathscr{M}$ satisfy
\[\lambda_2 - \lambda_1 =\mathscr{O}(e^{-(1-\kappa)\frac{S}{h}})\,.\]	
\end{proposition}

\begin{proof}
Note that
\[(\mathscr{M}-\mu) f_u =(\mathscr{M}-\mu) (\chi_u \mathbf f_u)=(\mathscr{M}_u -\mu) (\chi_u \mathbf f_u)\,,\]
where we used that $\Sigma(x_2)\chi_u(x_2)=0$. Then, we write
\[(\mathscr{M}_u-\mu) (\chi_u \mathbf f_u)=(\mathscr{M}_u - \mu) \mathbf f_u+(\mathscr{M}_u-\mu) ((\chi_u-1)\mathbf f_u)\,.\]
Since $\mathbf f_u$ is an eigenfunction of $\mathscr{M}_u$, we deduce
\[\|(\mathscr{M}-\mu) f_u\|=\|(\mathscr{M}_u-\mu) ((\chi_u-1)\mathbf f_u)\|\,.\]
Thanks to the Calder\'on-Vaillancourt theorem and Theorem \ref{thm.agmon}, we get that
\[\|(\mathscr{M}_u-\mu) ((\chi_u-1)\mathbf f_u)\|\leq C\| (\chi_u-1)\mathbf f_u\|_{B^2 \otimes L^2}=\mathscr{O}(e^{-(1-\kappa) \frac S h})\|\mathbf f_u\|\,,\]
where we used that $\tilde \varphi_u(x_2) > (1-\kappa) S$ on the support of $1-\chi_u$, if this support is small enough around $c_d$ (see \eqref{eq.phius}). 
Therefore, by using Theorem \ref{thm.agmon} again,
\[\|(\mathscr{M}-\mu) f_u\|={\mathscr{O}}(e^{-(1-\kappa) \frac S h})\|f_u\|\,.\]
By symmetry we also have
\[\|(\mathscr{M}-\mu) f_d\|={\mathscr{O}}(e^{-(1-\kappa) \frac S h})\|f_d\|\,.\]
Then, the spectral theorem tells us that there are two eigenvalues of $\mathscr{M}$ close to $\mu$, at a distance of order at most $\mathscr{O}(e^{-(1-\kappa)\frac Sh})$. Due to the a priori estimate \eqref{eq.apriorigap} of the spectrum of the double-well operator, these two eigenvalues are the only ones that close to $\mu$ and they are the lowest two ones. 
\end{proof}

We then construct a basis of the eigenspace $E$ using $f_d$ and $f_u$. Let us consider the projections of our one-well quasimodes $g_\star=\Pi_E f_\star$, with $\star\in\{u,d\}$. In order to orthonormalize this basis, we consider the Gram matrix,
\[G=\begin{pmatrix}
	\langle g_u,g_u\rangle&\langle g_u,g_d\rangle\\
	\langle g_d,g_u\rangle&\langle g_d,g_d\rangle
\end{pmatrix}=\begin{pmatrix}
	g_u\\
	g_d
\end{pmatrix}\cdot\begin{pmatrix}
	g_u&g_d
\end{pmatrix}\geq 0\,.\]
The matrix of the quadratic form associated with $\mathscr{M}$ in the basis $(g_u,g_d)$ is
\[M=\begin{pmatrix}
	\langle \mathscr{M} g_u,g_u\rangle&\langle  \mathscr{M} g_u,g_d\rangle\\
	\langle  \mathscr{M} g_d,g_u\rangle&\langle  \mathscr{M} g_d,g_d\rangle
\end{pmatrix}.\]
The following proposition is a consequence of Proposition \ref{prop.quasimodesexp}. It tells us that $\mathrm{span}(f_u, f_d)$ is an exponentially good approximation of $E$.

\begin{proposition}\label{prop.interaction}
	Let $\kappa >0$. We have, with $\star\in\{u,d\}$,
	\begin{equation}\label{eq.approx}
		\|g_\star-f_\star\|={\mathscr{O}}(e^{-(1-\kappa)\frac{S}{h}})\,,\quad \langle \mathscr{M} (g_\star-f_\star), g_\star - f_\star \rangle ={\mathscr{O}}(e^{-(1-\kappa)\frac{2S}{h}})\,.
	\end{equation}
	Moreover, the family $(g_u,g_d)$ is asymptotically an orthonormal basis of $F$. More precisely, we have
	\begin{equation}\label{eq.G}
	G=\mathrm{Id}+\mathrm{T}+{\mathscr{O}}(e^{-(1-\kappa) \frac{2S}{h}})\,,\quad \mathrm{T}=\begin{pmatrix}
			0&\langle f_u, f_d\rangle\\
			\langle f_d, f_u\rangle&0
		\end{pmatrix}\,.
	\end{equation}
	In addition, we have
	\begin{equation}\label{eq.L}
		M=\begin{pmatrix}
			\mu& w\\
			\overline{w}&\mu
		\end{pmatrix}+\mu\mathrm{T}+{\mathscr{O}}(e^{-(1-\kappa) \frac{2S}{h}})\,,\quad w=\langle(\mathscr{M}-\mu)f_u,f_d\rangle\,.
	\end{equation}
	Finally, the matrix of $\mathscr{M}$ in the orthonormal basis 
	$G^{-\frac12}\begin{pmatrix}
		g_u\\
		g_d
	\end{pmatrix}$
	is of the form
\[\begin{pmatrix}
		\mu& w\\
		\overline{w}&\mu
	\end{pmatrix}+{\mathscr{O}}(e^{-(1-\kappa) \frac{2S}{h}})+\mathscr{O}(|\rm T|^2)\,.\]
\end{proposition}

The proof of Proposition \ref{prop.interaction} follows standard arguments going back at least to Helffer-Sjöstrand~\cite{HS84}. We refer to \cite[Section 4.1]{FMR23} where identical estimates are obtained. Finally, the error term $|\rm T|^2$ can be estimated using the decay of $f_u$ and $f_d$.

\begin{lemma}\label{lem.T}
Let $\kappa >0$. If $\|\Sigma\|_\infty$ and $|{\rm{supp}} \,\Sigma|$ are small enough we have
\[ |{\rm T}| = \mathscr{O}(e^{-(1-\kappa) \frac Sh}). \]
\end{lemma}

\begin{proof}
We have to control the following scalar product,
\[ \langle f_u, f_d \rangle = \langle \chi_u \mathbf f_u, \chi_d \mathbf f_d \rangle. \]
We use the weight $\varphi_u$ defined in Lemma~\ref{lem.Brond=b0} and the corresponding version $\varphi_d$ in the other well, defined by $\varphi_d:=U \varphi_u$.
We insert the exponential weight $(1-\kappa)\varphi_u$ to get
\[ \langle f_u, f_d \rangle = \langle e^{-(1-\kappa)(\varphi_u + \varphi_d)} e^{(1-\kappa)\varphi_u} \chi_u \mathbf f_u,  e^{(1-\kappa)\varphi_d} \chi_d \mathbf f_d \rangle.\]
On the support of $\chi_u \chi_d$, we have $\varphi_u(x) + \varphi_d(x) = \varphi_u(x) + \varphi_u(-x) = S$ (see \eqref{eq.phius}), and $\varphi_u = \tilde \varphi_u$, with notation from \eqref{eq:defphitilde}. Thus, we can use the exponential decay from Theorem \ref{thm.agmon} and we find
\[ | \langle f_u, f_d \rangle | \leq C e^{-(1-\kappa)\frac Sh} \| \mathbf f_u \|^2,\]
which concludes the proof.
\end{proof}
\subsection{Estimate of the interaction and conclusion}\label{sec.analysisw}
It follows from Proposition~\ref{prop.interaction} (and Lemma \ref{lem.T}) that the first two eigenvalues of the double well operator $\mathscr{M}$ satisfy
\begin{equation} \label{eq.lambda}
 \lambda_1(\mathscr{M}) = \mu - |w| + \mathscr{O} \big(e^{-(1-\kappa) \frac{2S}{h}} \big), \quad  \lambda_2(\mathscr{M}) = \mu + |w| + \mathscr{O} \big(e^{-(1-\kappa) \frac{2S}{h}} \big). 
 \end{equation}
Therefore, the spectral gap is asymptotically given by $2|w|$. We now explain how to estimate $w$ using the WKB-like construction, Theorem \ref{thm.WKB}. We introduce the conjugate operator $\mathscr{M}^{\varphi_u}$ and we write
\[w=\langle (\mathscr{M}-\mu)f_u,f_d\rangle=\langle e^{-(\varphi_d+\varphi_u)/h}(\mathscr{M}^{\varphi_u}-\mu)\chi_u e^{\varphi_u/h}\mathbf f_u,\chi_d e^{\varphi_d/h} \mathbf f_d \rangle\,,\]
where $\varphi_d(x_2) = \varphi_u(-x_2)$. Note that $\varphi_u + \varphi_d$ is constant equal to $S$ on ${\rm{supp}} (\chi_u\chi_d) $, see \eqref{eq.phius} and \eqref{def.chiu}. Therefore, if we introduce $\underline{\chi}_u$ which is equal to $1$ on ${\rm{supp}}(\chi_u)$ and supported on $\R \setminus {\rm{supp}} (\Sigma)$ we have
\[w= e^{-S/h} \langle \underline{\chi}_u (\mathscr{M}^{\varphi_u}-\mu)\chi_u e^{\varphi_u/h} {\mathbf f}_u,\chi_d e^{\varphi_d/h} \mathbf f_d \rangle + \mathcal E,\]
with
\begin{align*}
\mathcal E &= \langle e^{-(\varphi_d+\varphi_u)/h}(1-\underline{\chi}_u)(\mathscr{M}^{\varphi_u}-\mu)\chi_u e^{\varphi_u/h}\mathbf f_u,\chi_d e^{\varphi_d/h} \mathbf f_d \rangle \\
&\leq e^{-S/h} |\langle (1-\underline{\chi}_u)(\mathscr{M}^{\varphi_u}-\mu)\chi_u e^{\varphi_u/h}\mathbf f_u,\chi_d e^{\varphi_d/h} \mathbf f_d \rangle| \\
&\leq e^{-S/h} \mathscr{O}(h^\infty),
\end{align*}
where we used $\varphi_u + \varphi_d \geq S$ and the pseudo-locality of the pseudodifferential operator $\mathscr{M}^{\varphi_u}$.
Hence
\begin{equation}\label{eq.w0}
w= e^{-S/h} \omega + \mathscr{O}\big( h^\infty e^{-S/h}\big),
\end{equation}
with
\[\omega=\langle \underline{\chi}_u (\mathscr{M}^{\varphi_u}-\mu) \chi_u e^{\varphi_u/h}\mathbf f_u,\chi_de^{\varphi_d/h}\mathbf f_d\rangle\,.\]
Using Theorem \ref{thm.WKB}, and since the quadratic form is controlled by a $B^2 \otimes L^2$-norm, we can replace $\mathbf f_u$ and $\mathbf f_d$ by their WKB approximations. Thus,
\begin{multline*}
\omega = \sqrt{\frac{\varphi_u''(c_u)}{h}} \big\langle \underline{\chi}_u (\mathscr{M}^{\varphi_u}-\mu) \chi_u \mathscr{Q}_+^{[J]}(\chi  a^{[J]}), \chi_d U \mathscr{Q}_+^{[J]}(\chi  a^{[J]}) \big\rangle \big(1+ \mathscr{O}(h^{\frac 12}) \big)\\+\mathscr{O}\big(h^{\frac J2 - 1 - N}\big)\,,
\end{multline*}
where $U$ is the symmetry $U f(x) = f(-x)$.
Note that for any fixed $N$, we can choose $J$ as large as we want, so that the error is small. Since $\Sigma = 0$ on ${\rm{supp}}( \chi_u )$, we can replace $\mathscr{M}^{\varphi_u}$ by $\mathscr{M}^{\varphi_u}_u$. Commuting and using Theorem \ref{thm.WKB} again we deduce that
\begin{multline*}
 \omega= \sqrt{\frac{\varphi_u''(c_u)}{h}}  \big\langle \underline{\chi}_u [ \mathscr{M}_u^{\varphi_u}, \chi_u] \mathscr{Q}_+^{[J]} ( \chi a^{[J]}), \chi_d U \mathscr{Q}_+^{[J]}(\chi  a^{[J]}) \big\rangle \big(1+ \mathscr{O}(h^{\frac 12}) \big) \\ +\mathscr{O}(h^{\frac J2 - 1 - N})\,. 
 \end{multline*}
The above commutator is an $h$-pseudodifferential operator with respect to $x_2$, such that
\[ [ \mathscr{M}_u^{\varphi_u}, \chi_u] \mathscr{Q}_+^{[J]} = \Op_h^{w,2} \Big( \frac hi \lbrace M_0, \chi_u \rbrace \mathsf{F}_{(x_2,\xi_2)} + \mathscr{O}\big(h^{\frac 32}\big)\Big),\]
where we recall that $M_0$ is the principal operator-symbol of $\mathscr{M}_u^{\varphi_u}$. Therefore, the stationary phase method gives
\begin{equation}
\omega =\sqrt{\varphi_u''(c_u) h} \int_{\R} \chi_u'(x_2)  {\mathcal J}(x_2)  a_0(x_2) \overline{a_0}(-x_2)\dd x_2 + \mathscr{O}(h),
\end{equation}
where we introduced the function
\begin{align}\label{def.Jcal}
{\mathcal J}(x_2) &= \int_{\R} \Big(  \frac 1i \partial_{\xi_2} M_0(x_2,0) \mathsf{F}_{(x_2,0)} \Big) (x_1) \, \overline{\mathsf{F}_{(-x_2,0)}}(-x_1) \dd x_1 .
\end{align}
This shows, using \eqref{eq.lambda} and \eqref{eq.w0}, that
\begin{equation}\label{eq.spectralgap1}
\lambda_2(\mathscr{M}) - \lambda_1(\mathscr{M}) = c_0 h^{\frac 12} e^{- S/h} \Big(1 + \mathscr{O} \big(h^{\frac 12}\big)\Big),
\end{equation}
with 
\begin{equation}\label{eq.c0a}
c_0 = 2 \sqrt{\varphi_u''(c_u)} \Big| \int_{\R} \chi_u'(x_2) \mathcal{J}(x_2) a_0(x_2) \overline{a_0}(-x_2) \dd x_2 \Big|.
\end{equation}
The double-well magnetic Laplacian $\mathscr{L}$ is unitarily equivalent to $h \mathscr{M}$ (Proposition~\ref{prop.Mh}), meaning that $\lambda_j(\mathscr{L}) = h \lambda_j(\mathscr{M})$, and thus Theorem~\ref{thm.main} is a reformulation of \eqref{eq.spectralgap1}. 

\begin{remark} Equation \eqref{eq.spectralgap1} gives 
\begin{equation}\label{eq.specgap0}
\lim_{h\to 0}e^{S/h}h^{-\frac12}(\lambda_2-\lambda_1)= c_0.
\end{equation}
However, the left-hand side in \eqref{eq.specgap0} is independent of the choice of cutoff $\chi_u$, which appears in $c_0$. We infer that the function $x_2\mapsto {\mathcal J}(x_2)a_{0}(x_2)\overline{a}_{0}(-x_2)$ must be constant on a small interval of the form $[c_d+\eta_0,c_d+2\eta_0]$.
Therefore, since ${\mathcal J}$ and $a_0$ do not vanish on $(c_d,c_u)$, we deduce that $c_0 \neq 0$. In fact, we show in Appendix \ref{sec.calculs} that ${\mathcal J}(x_2)a_{0}(x_2)\overline{a}_{0}(-x_2)$ is indeed constant everywhere, see Lemma \ref{lem.constant}. This also gives an explicit formula for $c_0$, and finishes the proof of Theorem~\ref{thm.main}.
\end{remark}

\appendix

\section{Calculation of the constant}\label{sec.calculs}

This appendix contains the calculations required to find a formula for the coefficient $c_0$ in the asymptotics \eqref{eq.spectralgap1} of the spectral gap. The final formula is given by Lemma \ref{lem.constant} below. 

We first calculate the quantity $\mathcal{J}(x_2)$ introduced in \eqref{def.Jcal}, which can be written in terms of the $L^2(\R_{x_1})$ scalar product,
\[ \mathcal{J}(x_2) = \frac 1i \langle \partial_{\xi_2} M_0(0,x_2) \mathsf{F}_{(x_2,0)} , \mathsf{F}_{(-x_2,0)} \rangle,\]
using that $\mathsf{F}_{(-x_2,0)}$ is even as a function of $x_1$ by the explicit formula \eqref{def.FX2}.
\begin{lemma}\label{lem.J}
For all $x_2 \in (c_d,c_u)$ such that $x_2 \notin {\rm supp}(\Sigma)$ and $-x_2 \notin {\rm{supp}} (\Sigma)$ we have
\[ \mathcal J(x_2) = \langle \mathsf{F}_{(-x_2,0)}, \mathsf{F}_{(x_2,0)} \rangle \frac 1i \partial_1 \mathcal B(i \varphi'_u(x_2),x_2). \]
\end{lemma}

\begin{proof}
This follows from a Feynman-Helmann argument. We recall the symmetry properties of $\alpha$ and $\mathcal B$. It follows from our assumptions that $\alpha(x_1,x_2) = - \alpha(x_1,-x_2)$ and 
\begin{equation*}
\begin{cases}
\mathcal B(-\xi_2 +  i \varphi_u'(-x_2),-x_2) = \overline{\mathcal B(\xi_2 + i \varphi_u'(x_2),x_2)}, \\
\mathcal \alpha(-\xi_2 +  i \varphi_u'(-x_2),-x_2) = \overline{\mathcal \alpha(\xi_2 + i \varphi_u'(x_2),x_2)},
\end{cases}
\end{equation*}
when $x_2$ satisfies our assumptions and $\xi_2 \in \R$.
Therefore, with the notation $X_2 = (x_2,\xi_2)$, and keeping in mind the normalization constants $C(X_2)$ and $c(X_2)$ in \eqref{def.FX2} and \eqref{def.GX2}, we have
\begin{align}\label{eq:SymmEigenfunction}
\mathsf F_{-X_2} = \frac{C(-X_2)}{c(X_2)} \mathsf G_{X_2} = \langle \mathsf F_{-X_2}, \mathsf  F_{X_2} \rangle \mathsf G_{X_2}.
\end{align}
The Gaussian $\mathsf{F}_{X_2}$ is an eigenfunction of $M_0(X_2)$ associated to the eigenvalue $\mathcal B(\xi_2 + i \varphi_u'(x_2),x_2) + \Sigma(x_2)$.
By differentiating the eigenvalue equation, we get
\begin{align*}
&\left( M_0(x_2,\xi_2) - \mathcal B( \xi_2 + i \varphi_u'(x_2),x_2) - \Sigma(x_2) \right) \partial_{\xi_2} \mathsf F_{(x_2,\xi_2)} \\
&\qquad= 
\left(- \partial_{\xi_2} M_0(x_2,\xi_2) +
\partial_{1}  \mathcal B( \xi_2 + i \varphi_u'(x_2),x_2) \right) \mathsf F_{(x_2,\xi_2)}.
\end{align*}
Therefore, with the $L^2({\mathbb R}_{x_1})$ scalar product, we have by the eigenfunction properties of $\mathsf G$,
\begin{align*}
\langle \partial_{\xi_2} M_0(x_2,\xi_2) \mathsf F_{(x_2,\xi_2)}, \mathsf G_{(x_2,\xi_2)} \rangle &= 
\langle  \partial_{1}  \mathcal B( \xi_2 + i \varphi_u'(x_2),x_2)  \mathsf F_{(x_2,\xi_2)}, \mathsf G_{(x_2,\xi_2)} \rangle \\
&= \partial_{1}  \mathcal B( \xi_2 + i \varphi_u'(x_2),x_2) .
\end{align*}
In particular, we get using \eqref{eq:SymmEigenfunction}
\begin{align}\label{eq:FeynmanHelmann}
\langle \partial_{\xi_2} M_0(X_2) \mathsf F_{X_2}, \mathsf F_{-X_2} \rangle =\langle \mathsf F_{-X_2}, \mathsf  F_{X_2} \rangle
\partial_{1}  \mathcal B(\xi_2 + i \varphi_u'(x_2),x_2),
\end{align}
from which the formula for $\mathcal J$ follows.
\end{proof}

The operator-symbol $\mathsf{Q}_2$ of $\mathscr{Q}_2$ (given in Proposition \ref{prop.parametrix}) is
\begin{equation}\label{eq.Q2b}
 \mathsf{Q}_2 = -\mathsf{Q}_0\mathsf{P}_2\mathsf{Q}_0+\mathsf{Q}_0\mathsf{P}_1\mathsf{Q}_0\mathsf{P}_1\mathsf{Q}_0-\frac{1}{2i}\{\mathsf{Q}_0,\mathsf{P}_0\}\mathsf{Q}_0\,,
\end{equation}
where $\mathsf{Q}_0 = \mathsf{P}_0^{-1}$ is given in \eqref{eq.P0-1} and $\{\mathsf{Q}_0,\mathsf{P}_0\} = \partial_{\xi_2} \mathsf{Q}_0 \partial_{x_2} \mathsf{P}_0 - \partial_{x_2} \mathsf{Q}_0 \partial_{\xi_2} \mathsf{P}_0$. We write this operator-symbol in the matrix form,
\[ \mathsf{Q}_2 = \begin{pmatrix}
\mathsf{Q}_2^{++} & \mathsf{Q}_2^+ \\ \mathsf{Q}_2^- & \mathsf{Q}_2^{\pm}
\end{pmatrix}.\]
Calculating $\mathsf{Q}_2$ using formula \eqref{eq.Q2b}, we find
\begin{align*}
\mathsf{Q}_2^\pm&(X_2) = - \langle M_2(X_2) \mathsf F_{X_2} , \mathsf{G}_{X_2} \rangle + \langle M_1(X_2) R_0(X_2) M_1(X_2) \mathsf{F}_{X_2}, \mathsf{G}_{X_2} \rangle \\
& - \frac{1}{2i} \Big( \langle \partial_{\xi_2} M_0 \partial_{x_2} \mathsf{F}_{X_2} , \mathsf{G}_{X_2} \rangle - \partial_{x_2} E \langle \partial_{\xi_2} \mathsf F_{X_2}, \mathsf{G}_{X_2} \rangle + (z- E) \langle \partial_{x_2} \mathsf{F}_{X_2}, \partial_{\xi_2} \mathsf{G}_{X_2} \rangle \Big) \\
&+\frac{1}{2i} \Big( \langle \partial_{x_2} M_0 \partial_{\xi_2} \mathsf{F}_{X_2} , \mathsf{G}_{X_2} \rangle - \partial_{\xi_2} E \langle \partial_{x_2} \mathsf F_{X_2}, \mathsf{G}_{X_2} \rangle + (z- E) \langle \partial_{\xi_2} \mathsf{F}_{X_2}, \partial_{x_2} \mathsf{G}_{X_2} \rangle \Big),
\end{align*}
with $E(x_2,\xi_2) = \mathcal{B}(\xi_2 + i \varphi_u'(x_2),x_2) + \Sigma(x_2)$. For this specific choice of weight $\varphi_u$, we have $E(x_2,0) = b_0$, and $\partial_{x_2} E (x_2,0) = 0$. Thus, when $z=b_0$ we find the simpler formula
\begin{align*}
\mathsf{Q}_2^\pm(x_2,0) = &- \langle M_2(X_2) \mathsf F_{X_2} , \mathsf{G}_{X_2} \rangle + \langle M_1(X_2) R_0(X_2) M_1(X_2) \mathsf{F}_{X_2}, \mathsf{G}_{X_2} \rangle \\
& - \frac{1}{2i}  \langle \partial_{\xi_2} M_0 \partial_{x_2} \mathsf{F}_{X_2} , \mathsf{G}_{X_2} \rangle + \frac{1}{2i} \langle \partial_{x_2} M_0 \partial_{\xi_2} \mathsf{F}_{X_2} , \mathsf{G}_{X_2} \rangle \\
& - \frac 1 {2i} \partial_1 \mathcal B(i \varphi_u'(x_2),x_2) \langle \partial_{x_2} \mathsf{F}_{X_2}, \mathsf{G}_{X_2} \rangle
\end{align*}
with $X_2=(x_2,0)$. Note that $(M_0- E)^* \mathsf G = 0$ and thus
\[\partial_{\xi_2} M_0^* \mathsf G = \overline{ \partial_{\xi_2} E} \mathsf G - (M_0 - E)^* \partial_{\xi_2} \mathsf G, \]
which then gives
\begin{align}
\mathsf{Q}_2^\pm(x_2,0) = &- \langle M_2(X_2) \mathsf F_{X_2} , \mathsf{G}_{X_2} \rangle + \langle M_1(X_2) R_0(X_2) M_1(X_2) \mathsf{F}_{X_2}, \mathsf{G}_{X_2} \rangle \nonumber \\
& - \frac{1}{2i} \langle (M_0 - E) \partial_{\xi_2} \mathsf F_{X_2}, \partial_{x_2} \mathsf{G}_{X_2} \rangle + \frac{1}{2i} \langle (M_0 - E) \partial_{x_2} \mathsf F_{X_2}, \partial_{\xi_2} \mathsf{G}_{X_2} \rangle \nonumber \\ 
& - \frac 1 i \partial_1 \mathcal B(i \varphi_u'(x_2),x_2) \langle \partial_{x_2} \mathsf{F}_{X_2}, \mathsf{G}_{X_2} \rangle. \label{eq.Q2exp}
\end{align}

\begin{lemma}\label{lem.sym.Q2}
The function $\mathsf{Q}_2^\pm$, defined with weight $\varphi_u$ and $z=b_0$ satisfies,
\[ \mathsf{Q}_2^\pm(x_2,0) - \overline{\mathsf{Q}_2^\pm(-x_2,0)} = i \partial_{1} \mathcal B(i\varphi_u'(x_2),x_2) \partial_{x_2} \big( \ln \langle \mathsf F_{(-x_2,0)}, \mathsf{F}_{(x_2,0)} \rangle \big), \]
for all $x_2 \in (c_d,c_u)$ such that $x_2 \notin {\rm{supp}} (\Sigma)$ and $-x_2 \notin{\rm{supp}} (\Sigma)$.
\end{lemma}

\begin{proof}
We note the following symmetry properties, for $X_2 =(x_2,\xi_2)$,
\[ M_0(-X_2) = M_0^*(X_2), \quad M_1(-X_2) = -M_1^*(X_2), \quad M_2(-X_2)=-M_2^*(X_2), \]
which implies $R_0(-X_2) = R_0(-X_2)^*$, and also
\[ \mathsf F_{-X_2} = \langle \mathsf F_{-X_2}, \mathsf F_{X_2} \rangle \mathsf{G}_{X_2}, \qquad \mathsf G_{-X_2} = \langle \mathsf F_{-X_2}, \mathsf F_{X_2} \rangle^{-1} \mathsf{F}_{X_2},\]
which give
\begin{align*}
\partial_{x_2} \mathsf F (-X_2) &= - \partial_{x_2} (\langle \mathsf F_{-X_2}, \mathsf F_{X_2} \rangle ) \mathsf G_{X_2} - \langle \mathsf F_{-X_2}, \mathsf F_{X_2} \rangle \partial_{x_2} \mathsf G(X_2) \\
\partial_{x_2} \mathsf G (-X_2) &= - \partial_{x_2} (\langle \mathsf F_{-X_2}, \mathsf F_{X_2} \rangle^{-1} ) \mathsf F_{X_2} - \langle \mathsf F_{-X_2}, \mathsf F_{X_2} \rangle^{-1} \partial_{x_2} \mathsf F(X_2),
\end{align*}
and similar formulas for the $\xi_2$-derivatives. Using these symmetries we find, when $\xi_2=0$,
\begin{equation*}
\begin{cases}
\langle M_1(-X_2) R_0(-X_2) M_1(-X_2) \mathsf{F}_{-X_2}, \mathsf{G}_{-X_2} \rangle = \overline{\langle M_1(X_2) R_0(X_2) M_1(X_2) \mathsf F_{X_2}, \mathsf{G}_{X_2} \rangle}\\
\langle M_2(-X_2) \mathsf F_{-X_2}, \mathsf{G}_{-X_2} \rangle = \overline{\langle M_2(X_2) \mathsf{F}_{X_2}, \mathsf{G}_{X_2} \rangle}\\
\langle (M_0(-X_2) - E(-X_2)) \partial_{\xi_2} \mathsf{F}_{-X_2}, \partial_{x_2} \mathsf{G}_{-X_2} \rangle = \overline{ \langle (M_0 - E) \partial_{x_2} \mathsf F_{X_2}, \partial_{\xi_2} \mathsf{G}_{X_2} \rangle}\\
\langle (M_0(-X_2) - E(-X_2)) \partial_{x_2} \mathsf{F}_{-X_2}, \partial_{\xi_2} \mathsf{G}_{-X_2} \rangle = \overline{ \langle (M_0 - E) \partial_{\xi_2} \mathsf F_{X_2}, \partial_{x_2} \mathsf{G}_{X_2} \rangle}
\end{cases}
\end{equation*}
and, most importantly,
\begin{align*}
 \partial_1 \mathcal B(i\varphi_u'(-x_2),x_2) \langle \partial_{x_2} \mathsf F_{-X_2}, \mathsf G_{-X_2} \rangle =&- \partial_{1} \mathcal B(i\varphi_u'(x_2),x_2) \partial_{x_2} \big( \ln \langle \mathsf F_{(-x_2,0)}, \mathsf{F}_{(x_2,0)} \rangle \big) \\ & - \overline{ \partial_1 \mathcal{B}(i\varphi_u'(x_2),x_2) \langle \partial_{x_2} \mathsf F_{X_2}, \mathsf G_{X_2} \rangle} .
 \end{align*}
 The Lemma is then proven using formula \eqref{eq.Q2exp}.
\end{proof}

\begin{lemma}\label{lem.constant}
Let $a_0$ be the solution to the transport equation \eqref{def.a0}. Then, for $x_2 \in (c_d,c_u)$ such that $x_2 \notin {\rm{supp}} (\Sigma)$ and $-x_2 \notin{\rm{supp}} (\Sigma)$, we have 
\[a_0(x_2) \overline{a_0}(-x_2) \mathcal J(x_2) = \exp \Big( \int_{c_d}^{c_u} \frac{\overline{D(-s)}}{i \partial_1 \mathcal B(i \varphi_u'(s),s)} \dd s \Big),\]
where $D$ is defined in \eqref{eq.D}. In particular, the constant $c_0$ in \eqref{eq.c0a} is equal to
\begin{equation}\label{eq.c0b}
 c_0 = 2 \sqrt{\varphi_u''(c_u)} \exp \Big( \int_{c_d}^{c_u} \frac{\overline{D(-s)}}{i \partial_1 \mathcal B(i \varphi_u'(s),s)} \dd s \Big).
 \end{equation}
\end{lemma}

\begin{proof}
We use the formula \eqref{def.a0} defining $a_0$. In particular,
\begin{align*}
a_0(-x_2) &= \exp \Big( -\int_{c_u}^{-x_2} \frac{D(s)}{i \partial_1 \mathcal B(\varphi_u'(s),s)} \dd s  \Big) \\
&= \exp \Big( \int_{c_d}^{x_2} \frac{D(-s)}{i \partial_1 \mathcal B(\varphi_u'(s),s)} \dd s  \Big) \\
&= \exp \Big( \int_{c_d}^{c_u} \frac{D(-s)}{i \partial_1 \mathcal B(\varphi_u'(s),s)} + \int_{c_u}^{x_2} \frac{D(-s)}{i \partial_1 \mathcal B(\varphi_u'(s),s)} \dd s \Big).
\end{align*}
Since $i \partial_1 \mathcal B(i \varphi_u'(s),s)$ is real we deduce that
\begin{equation}\label{eq.a0a01}
a_0(x_2) \overline{a_0}(-x_2) = K \exp \Big( - \int_{c_u}^{x_2} \frac{D(s)-\overline{D(-s)}}{i \partial_1 \mathcal B(\varphi_u'(s),s)} \dd s \Big),
\end{equation}
with 
\begin{equation*}
K = \exp \Big( \int_{c_d}^{c_u} \frac{\overline{D(-s)}}{i \partial_1 \mathcal B(\varphi_u'(s),s)} \dd s \Big).
\end{equation*}
We recall that
\[D(x) =  g(x) + \mathsf{Q}_2^\pm(x,0) + z_2,\]
where $z_2 \in \R$ and
\[g(x) = \frac 12 \frac{\dd}{\dd x} \Big( \partial_1 \mathcal B(i \varphi_u'(x),x) \Big).\]
Since $\overline{g(-x)} = -g(x)$ and using Lemma \ref{lem.sym.Q2} we find
\[\frac{D(x) - \overline{D(-x)}}{i \partial_{1} \mathcal B(i \varphi_u'(x),x)} = \frac{\dd }{\dd x} \Big( \ln \big( i \partial_1 \mathcal B(i \varphi_u'(x),x) \langle \mathsf F_{(-x_2,0)}, \mathsf F_{(x_2,0)} \rangle \Big). \]
Equation \ref{eq.a0a01} then becomes
\[ a_0(x_2) \overline{a_0}(-x_2) = K \mathcal J(x_2)^{-1}, \]
where we used Lemma \ref{lem.J}. 
\end{proof}

\section*{Acknowledgments}
This work was partially conducted within the France 2030 framework programme, the Centre Henri Lebesgue ANR-11-LABX-0020-01. N.R. is deeply grateful to the University of Copenhagen (and its QMATH group) where this work was started. This stay was partially funded by the IRN MaDeF (CNRS). N.R.  also thanks Antide Duraffour and Frédéric Hérau for enlightening discussions. S.F. and L.M. were partially supported by the grant 0135-00166B from the Independent Research Fund Denmark, by the VILLUM Foundation grant no. 10059, and by the ERC Advanced Grant MathBEC - 101095820.

\bibliographystyle{abbrv}
\bibliography{biblioGBFMR25}

\end{document}